\documentclass[10pt,journal,compsoc]{IEEEtran}
\include{micros}
\usepackage{times,amsmath,epsfig}
\usepackage{epstopdf}
\usepackage{multirow}
\usepackage{color}
\usepackage{amssymb}
\usepackage{graphicx}
\usepackage{amsthm}
\usepackage{cite}
\usepackage{url}
\usepackage{verbatim}
\usepackage[normalem]{ulem}
\usepackage{caption}
\usepackage{enumerate}
\usepackage{algorithmic}
\usepackage[table]{xcolor}
\usepackage{tabularx} % for 'tabularx' env. and 'X' col. type
\usepackage{ragged2e} % for \RaggedRight macro
\usepackage{booktabs} % for \toprule, \midrule etc macros
\usepackage{threeparttable}
\usepackage{float}
\usepackage[ruled,linesnumbered]{algorithm2e}
\usepackage{amssymb}
\usepackage{pifont}
\usepackage{multirow} % Required for multirows
\usepackage{subfig}
\usepackage{graphicx}

\newtheorem{definition}{Definition}
\newtheorem{theorem}{Theorem}
\newtheorem{remark}{Remark}
\begin{document}

\title{An Efficient and Multi-Private Key Secure Aggregation Scheme for Federated Learning}

\author{Xue~Yang,
        Zifeng~Liu, 
        Xiaohu~Tang,
        Rongxing~Lu,~\IEEEmembership{Fellow,~IEEE},
        and~Bo~Liu     
\thanks{X. Yang, Z. Liu, and X. Tang are with the Information Coding and Transmission Key Laboratory of Sichuan Province, Southwest Jiaotong University, Chengdu, China. E-mail: \{xueyang, xhutang\}@swjtu.edu.cn, lzf@my.swjtu.edu.cn.}% <-this % stops a space
\thanks{R. Lu is with the Canadian Institute of Cybersecurity, Faculty of Computer Science, University of New Brunswick,  Fredericton E3B 5A3, Canada. E-mail: rlu1@unb.ca.}
\thanks{B. Liu is with the DBAPPSecurity Ltd, Hangzhou, China. E-mail: bo.liu@dbappsecurity.com.cn.}
\thanks{Corresponding author: Xue~Yang (E-mail: xueyang@swjtu.edu.cn).}}

\markboth{Journal of \LaTeX\ Class Files,~Vol.~14, No.~8, August~2021}%
{Shell \MakeLowercase{\textit{et al.}}: A Sample Article Using IEEEtran.cls for IEEE Journals}


% \IEEEtitleabstractindextext{
% \begin{abstract}
% ada
% \end{abstract}

% \begin{IEEEkeywords}
% ada
% \end{IEEEkeywords}}

% \maketitle

% \IEEEdisplaynontitleabstractindextext

% \IEEEpeerreviewmaketitle

\IEEEtitleabstractindextext{%
\begin{abstract}
In light of the emergence of privacy breaches in federated learning,  secure aggregation protocols, which mainly adopt either homomorphic encryption or threshold secret sharing techniques, have been extensively developed for federated learning to preserve each client's local training data privacy. Nevertheless, many existing protocols suffer from several shortcomings such as dependence on a trusted third party, vulnerability to corrupted clients, low efficiency, and the trade-off between security and dropout-resiliency guarantee, etc. To deal with these shortcomings, we propose an efficient and multi-private key secure aggregation scheme for federated learning. Specifically, we skillfully design a multi-private key secure aggregation protocol that achieves homomorphic addition operation, with two important benefits: 1) both the server and each client can freely select public and private keys without introducing a trusted third party, and 2) the plaintext space is relatively large, making it more suitable for deep models. Besides, for dealing with the high dimensional deep model parameter, we introduce a super-increasing sequence to compress multi-dimensional data into one dimension, which greatly reduces encryption and decryption times as well as communication for ciphertext transmission. Detailed security analyses show that our proposed scheme can achieve semantic security of both individual local gradients and the aggregated result while achieving optimal robustness in tolerating client collusion and dropped clients. Extensive simulations demonstrate that the accuracy of our scheme is almost the same as the non-private approach, while the efficiency of our scheme is much better than the state-of-the-art homomorphic encryption-based secure aggregation schemes. More importantly, the efficiency advantages of our scheme will become increasingly prominent as the number of model parameters increases.
\end{abstract}

}

\maketitle

\IEEEdisplaynontitleabstractindextext
\IEEEpeerreviewmaketitle

\section{Introduction}
\IEEEPARstart{T}{he} continued emergence of privacy leaks and data abuse has hindered the flourishing of traditional centralized training, which collects a vast amount of training data from distributed data providers.
As highlighted in \cite{YangLCT19}, data providers are no longer comfortable uploading local data due to concerns over personal privacy and data control rights. Obviously, it is challenging to train a high-performance deep model without significant training data. To address privacy concerns, \emph{federated learning} \cite{KonecnyMYRSB16}, as one of the most important research aspects of private computing \cite{GaiWZXZ19}, has recently emerged. Federated learning is a distributed framework where many data providers (also called clients) collaboratively train a shared global model under the orchestration of a central server. During the phase of training, the training data are maintained locally, and clients only send the local gradient to the server. This training framework of federated learning addresses data abuse and significantly improves the data privacy of clients. As a result, the research on federated learning has grown significantly in recent years. 

Regrettably, recent works \cite{ShokriSSS17, HayesMDC19, ZhuLH19, GeipingBD020} have revealed that the adversary may still access some private information of training data or even reconstruct training data from interacted model parameters or local gradient. Specifically, \cite{ShokriSSS17} and \cite{HayesMDC19} investigate membership inference attacks based on the model parameter to infer if a particular data record was included in the training dataset. \cite{ZhuLH19} and \cite{GeipingBD020} demonstrate that adversaries may reconstruct training data from the local gradient. To improve the security of federated learning, many privacy-preserving federated learning schemes \cite{PhongAHWM18, HaoLXLY19, ChaiWCY21, ZhangLX00020, ZhangFWZC20, BonawitzIKMMPRS17, DongCSW20, XuLL0L20} have been presented. 
These schemes commonly employ two cryptographic techniques, i.e., homomorphic encryption \cite{Paillier99} and $(t, n)$-threshold secret sharing \cite{Shamir79}, to achieve secure aggregation and ensure privacy preservation of the local gradient for each client. However, these schemes suffer from several drawbacks that hinder their practical implementation: 
\begin{itemize}
    \item The schemes \cite{PhongAHWM18, PhongA0WM17, HaoLXLY19, ChaiWCY21, ZhangLX00020, ZhangFWZC20} that utilize the homomorphic encryption technique require all clients to share a pair of public and private keys generated by a trusted third party, and encrypt their local gradient with the same public key. Obviously, if an adversary compromises a client, the system is no longer secure, and meanwhile, no collusion is allowed in these schemes. Additionally, these schemes incur significant computational costs and communication overhead.
       
    \item The schemes \cite{BonawitzIKMMPRS17, DongCSW20, XuLL0L20} that use the $(t, n)$-threshold secret sharing technique have storage costs for private double-masks proportional to the number of clients. Moreover, for each iteration, all clients must re-select masked secrets and interactively perform the corresponding secret-sharing operation, resulting in additional interactive time and communication overhead. 
    Besides, there exists a trade-off between security (i.e.,  robustness against client collusion) and the dropout-resiliency guarantee, i.e., these secret sharing-based schemes increase the privacy guarantee by reducing the dropout-resiliency guarantee, and vice versa. Furthermore, these schemes do not consider preserving the privacy of the aggregated result or the well-trained model. Therefore, any adversary, in addition to the server, can also obtain the aggregated result once they have the transmitted data.
    
\end{itemize}

Hence, designing an efficient secure aggregation scheme that allows for client self-selected keys, robustness against collusion attacks, and tolerance for dropped clients remains a challenge. To overcome this challenge, we propose an efficient and multi-private key secure aggregation scheme for federated learning that supports homomorphic encryption with multi-private keys, decryption limitation for the aggregated result, robustness against collusion attacks, and tolerance for dropped clients The main contributions of this paper are threefold:

\begin{itemize}
    \item First, we address the privacy concerns of both the local gradient for each client and the aggregated result by skillfully designing a multi-private key secure aggregation protocol that does not require a trusted third party. This protocol allows the server and each client to select a pair of public and private keys freely. 
    Specifically, each client encrypts its local gradient with its own public key to ensure that the leakage of a particular client's private key does not compromise other clients' privacy. At the same time, even if up to $N-2$ clients collude with the server, they cannot access any information other than the aggregated result of the remaining two clients.
    Finally, only the server can decrypt the aggregated result with its private key.

    \item Second, to ensure efficiency, we employ a super-increasing sequence to greatly reduce the computational costs and communication overhead. Instead of encrypting each dimension of the multidimensional model parameters separately,  our scheme compresses multidimensional gradients into one dimension using this sequence before encryption. Obviously, this design significantly decreases the number of encryption and decryption operations (i.e., computational costs) and the number of transmitted ciphertexts (i.e., communication overhead). However, introducing the super-increasing sequence increases the message space length considerably. To this end, we skillfully design the encryption operation of our multi-private key secure aggregation scheme to support encrypting messages of large lengths.

    \item Detailed security analyses demonstrate that our scheme can ensure the semantic security of the local gradient for each client and the aggregated result, as well as robustness against collusion between the server and up to $N-2$ clients, and tolerance for up to $N-2$ dropped clients. Extensive experiments demonstrate our scheme exhibits significantly better communication and computational efficiency than the related secure aggregation work. Besides, the accuracy of our scheme is almost identical to the most popular federated learning scheme that does not consider privacy preservation. 
\end{itemize}

The rest of this paper is structured as follows. We outline  models and design goals in Section \ref{sec:model}. Then, we present our scheme in Section \ref{sec:scheme}, followed by its security analysis and performance evaluation in Sections \ref{sec:security} and \ref{sec:perfor}, respectively. Related work is discussed in Section \ref{sec:related}. Finally, we conclude our work in Section \ref{sec:conc}.

\section{Models and design goals}\label{sec:model}
This section commences by outlining the FL system model and corresponding threat model employed in this paper, followed by the identification of our design goals. Before proceeding with the detailed explanation, we provide a description of the notations used in the proposed scheme in Table Table \ref{tab:table1}.

\setlength{\arrayrulewidth}{0.3mm}
\setlength{\tabcolsep}{12pt}
\renewcommand{\arraystretch}{1.3}
{\rowcolors{1}{gray!20}{white}
\begin{table}[!t]
	\caption{Notation Used in the Proposed Scheme\label{tab:table1}}
	\centering
	\begin{tabular}{l l}
		\hline
		Notation & Description\\
		\hline
        $\{\mathcal{C}_1, \mathcal{C}_2, \ldots, \mathcal{C}_N\}$ & All clients in the system \\
        $\kappa_{1}$& Security parameter of system\\
        $\kappa_{2}$&  The bit length of the symmetric key\\
		$(p,g,q,\mathbb{G})$ & Parameters of system \\
        $H(\cdot)$& Cryptographic hash function\\
        $\mathbf{a}$ & A super-increasing vector\\
        $(\alpha, \beta=g^{\alpha})$ & Private and public keys of the server \\
        $(sk_i, pk_i=g^{sk_i})$& Private and public keys of $\mathcal{C}_i$ \\
        $key_{i}$& Symmetric key of the client $\mathcal{C}_i$\\
        $pk_{S}$& Aggregated public key\\
		$(r_{i1},r_{i2})$& Random number chosen by the client $\mathcal{C}_i$\\
        $W$& Global model parameters\\
		$\nabla_i$& Local gradients of $\mathcal{C}_i$ \\
		$\widehat{\nabla}_{i}$& Compressed gradients of $\mathcal{C}_i$ \\
		$E(\widehat{\nabla}_{i})=(E_{i1}, E_{i2})$& Encrypted gradient of $\mathcal{C}_i$ \\
        $W_{i}^{*}$& Encrypted global model parameter $W$\\
        $R$& A challenge generated by the server\\
        $(E_{Agg}, d, T)$& Aggregated result\\
		$\nabla$& Aggregated gradient \\
		\hline
	\end{tabular}
\end{table}

\subsection{System Model}\label{subsec:system_model}
As demonstrated in the majority of federated learning frameworks (e.g., \cite{0003FFSTXL22, HaoLXLY19, PachecoFSA18}), federated learning is essentially a distributed machine learning framework that enables clients to collaboratively train a global model under the orchestration of a central server, without exchanging the local training data of each client. 
Therefore, our system comprises two types of entities: a server and a number of clients $\{\mathcal{C}_{1}, \mathcal{C}_{2}, \ldots, \mathcal{C}_{N}\}$, each of which is responsible for executing the following operations: 
\begin{itemize}
    \item \emph{Server}: The server is accountable for aggregating the local gradients received from clients, updating the global model, and broadcasting the updated global model to clients for the next iteration.     
    \item \emph{Clients}: Each client $\mathcal{C}_{i}$ ($i\in [1, N]$) conducts local model training using the global model received from the server and the local training dataset to derive local gradients. Next, $\mathcal{C}_{i}$ uploads the local gradients to the server for aggregation. It is noteworthy that, due to the relatively large number of clients, we cannot ensure that all clients can participate in every iteration, particularly for mobile or IoT devices with unreliable connections. As a result, dropped clients are common \cite{LiSTS20}.    
\end{itemize}

Additionally, as depicted in Fig. \ref{fig_system}, the conventional framework of federated learning facilitates the server and clients to collaboratively execute the following two phases until the model converges:
\begin{enumerate}
    \item \emph{Local model training}: Initially, the server transmits the current global model parameter $W$ to all clients. Next, each client $\mathcal{C}_{i}$ calculates local gradients $\nabla_{i}$ using the received $W$ and local training data $\mathcal{D}_{i}$ by applying the stochastic gradient descent (SGD) algorithm. After that, $\mathcal{C}_{i}$ uploads the computed $\nabla_{i}$ to the server for the model update.
 
    \item \emph{Global model aggregation and update}: Initially, the server conducts weighted average aggregation on the received local gradients. Subsequently, the server updates the current model $W$ for the next iteration. Specifically, given the learning rate $\eta$, the aggregation and update operations are expressed as:    
    \begin{equation}\label{eq:aggre}
        W \Leftarrow \underbrace{W-\eta \underbrace{\sum_{i\in S_{a}} \frac{|\mathcal{D}_{i}|}{|\mathcal{D}|}\nabla_{i}}_{\mathrm{Aggregation}}}_{\mathrm{Update}}
    \end{equation}
\end{enumerate}

\begin{figure*}[!t]
	\centering
	\includegraphics[width=7in]{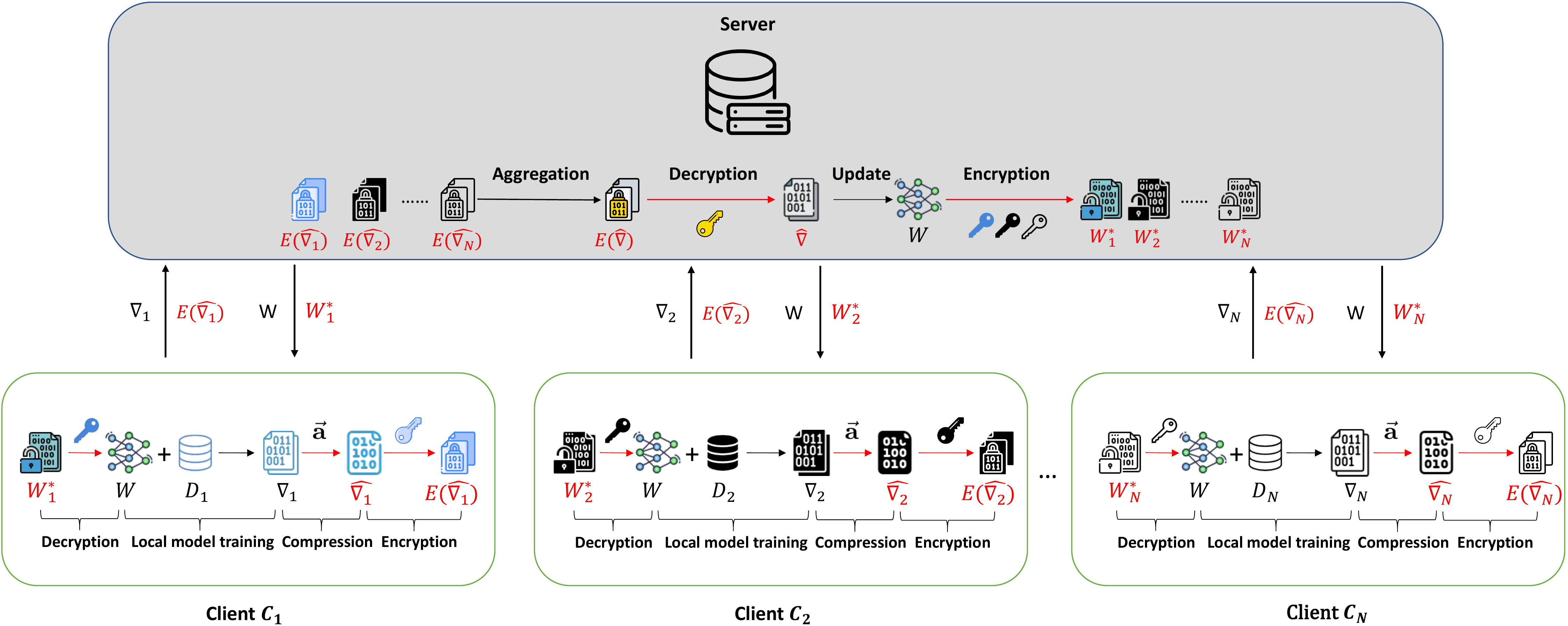}
	\caption{System model under consideration. Please note that the steps and parameters illustrated in black denote traditional federated learning, while those highlighted in red represent additional elements associated with privacy preservation.
 }\label{fig_system}
\end{figure*}

\subsection{Threat Model}\label{subsec:threat}
This paper considers two types of adversaries in the threat model, namely internal adversaries and external eavesdroppers. The attack capabilities corresponding to each type are defined as follows:
\begin{itemize}
    \item \emph{Internal adversary}: An internal adversary could either be the server or a client. Analogous to most privacy-preserving federated learning mechanisms \cite{XuBZAL19, XuLL0L20}, both the server and clients are assumed to be honest but curious, which means that they honestly follow the underlying scheme, but they are curious about the private training data of each client. More precisely, the server or client aims to extract the local gradients of a specific client to obtain the corresponding local training data through private attacks like reconstruction attacks \cite{ZhuLH19}. 
    Furthermore, to augment the attack capability, the server may collude with some clients to jointly obtain the local gradients of other clients. 
    Concretely, these colluded clients may share some private values with the hope of assisting in the acquisition of the local training data of a specific non-colluded client.

    \item \emph{External eavesdropper}: An external eavesdropper endeavors to acquire the local training data of a specific client by intercepting the contents exchanged between the server and clients.
\end{itemize}

\subsection{Design Goals}\label{subsec:design}
Drawing from the above system model and threat model, our proposed scheme aims to achieve the following objectives:
\begin{itemize}
    \item \emph{Privacy-preservation}: Neither the server nor the clients can access the private training data of any particular client, even if they collude. In particular, our scheme guarantees that neither the server nor the clients can acquire the local gradients of any specific client, even if some curious clients collude with the server. 
    Furthermore, another goal is to preserve the confidentiality of the aggregated result, i.e., the global model.
    More specifically, our scheme ensures that only the server can decrypt the aggregated gradient and obtain the well-trained model. Moreover,  external eavesdroppers cannot extract more private information than internal attackers.  
    
    \item \emph{Efficiency and robustness for tolerating dropped clients}:  In addition to the privacy-preserving requirement, efficiency, which includes computation, communication, and model accuracy, is one of the most critical metrics for federated learning \cite{WangWL19}, particularly for mobile or IoT devices\cite{KhanSHHH21}. 
    In reality, the implementation of privacy-preserving technology will undoubtedly lead to a loss in efficiency. Therefore, minimizing efficiency loss as much as possible is also one of our design goals. Additionally, as shown in \ref{subsec:system_model}, the scheme's robustness in tolerating dropped clients must be taken into account due to the system heterogeneity in federated learning.    
\end{itemize}

\section{Proposed scheme}\label{sec:scheme}
This section provides a detailed description of our secure aggregation scheme for federated learning. Our secure aggregation scheme mainly comprises three steps, in accordance with the conventional workflow of federated learning and the privacy-preserving objectives. These steps are as follows: 1) System parameter initialization, 2) Local model training and encryption, and 3) Model aggregation and decryption, which are highlighted in red in Fig. \ref{fig_system}. 
In essence, before commencing federated model training, the system parameters must be generated. Subsequently, each client $\mathcal{C}_{i}$ undertakes the local training under the orchestration of the server. 
More specifically, given the global model $W$, $\mathcal{C}_{i}$ computes the local gradient $\nabla_{i}$, and encrypts it using its private key $sk_{i}$. Subsequently, $\mathcal{C}_{i}$ uploads the encrypted gradient, denoted as $E(\nabla_{i})$, to the server. 
 The server then aggregates the received $E(\nabla_{i})$ and decrypts the aggregation result with the assistance of online clients. Following this, the server uploads the current model $W$ according to Eq. \eqref{eq:aggre}. The overall flow of our algorithm is demonstrated in Algorithm \ref{alg:alg2}, where the detailed operation of each step will be elucidated in the following sections.

\begin{algorithm}[htbp]
	\caption{The overall flow of our scheme}\label{alg:alg2}
		\textbf{System initialization: }{
                \begin{itemize}
                    \item \textbf{Server-side}: The server performs the system parameter generation function $\mathbf{SysGen}(\kappa_{1}, \kappa_{2})\rightarrow (p, q, \mathbb{G}, g, \beta, \mathbf{a}, H(\cdot))$\;
                    \item  \textbf{Server-side $\rightarrow$ Client-side}: The server sends the system parameter $\mathrm{SysPara}=(p, q, \mathbb{G}, g, \beta, \mathbf{a}, H(\cdot))$ to all clients\;
                    \item  \textbf{Client-side}: Each client $\mathcal{C}_{i}$ performs the key generation function $\mathbf{KeyGen}(p, q, \mathbb{G}, g, \beta, \mathbf{a})\rightarrow (pk_{i}, sk_{i})$\;
                    \item \textbf{Client-side $\rightarrow$ Server-side}: Each client $\mathcal{C}_{i}$ sends $pk_{i}$ to the server\;
                \end{itemize}}   
        \While{The global model $W$ does not converge}{
         %\textbf{Model training and encryption:} All clients in $S$ execute the following  operations in parallel, which is shown below\;
        \textbf{Server-side:} The server determines the set of online clients $S$ and conducts the aggregated public key generation function $\mathbf{AggPKGen}(pk_{i}, S)\rightarrow pk_{S}$ and the symmetric encryption function $\mathbf{SE}(pk_{i}, \alpha, H(\cdot), W)\rightarrow W_{i}^{*}$\;
        \textbf{Server-side $\rightarrow$ Client-side:} The server sends $(W_{i}^{*}, pk_{S})$ to each client $\mathcal{C}_{i}$ in $S$\;
        \textbf{Client-side:} $\mathcal{C}_{i}$ checks whether $pk_{S}\neq pk_{i}$ holds\;
        \If{$pk_{S}\neq pk_{i}$}{
         \For{$\mathcal{C}_{i} \in S$}{    
                    Perform the symmetric decryption operation $\mathbf{DSE}(\beta, H(\cdot), sk_{i}, W_{i}^{*})\rightarrow W$ \;
                    Train the model and obtain local gradients $\nabla_i$\;
                    Perform the encryption function $\mathbf{Enc}(\nabla_{i}, pk_{S}, \mathrm{SysPara})\rightarrow E(\widehat{\nabla}_{i})$\; 
                    Send the ciphertext $(E_{i1},E_{i2})$ to the server\;
                }          
        %\textbf{Model aggregation and decryption:} The server and all clients in $S$ interactively execute the following operations\;
        \textbf{Server-side:} The server performs the challenge generation function $\mathbf{ChelGen}(E(\widehat{\nabla}_{i}))\rightarrow R$\;
        \textbf{Server-side $\rightarrow$ Client-side:} The server sends $R$ to all clients in $S$\;
            \textbf{Client-side:} $\mathcal{C}_{i}\in S$ performs the response generation function $\mathbf{ResGen}(R, E(\widehat{\nabla}_{i}))\rightarrow T_{i}$\;  
            \textbf{Client-side $\rightarrow$ Server-side:} $\mathcal{C}_{i}$ sends $T_i$ to the server\;
            \textbf{Server-side:} The server performs         
            the aggregation function $\mathbf{Agg}(E(\widehat{\nabla}_{i}), T_{i})\rightarrow (E_{Agg}, d, T)$, and then executes the decryption function $\mathbf{Dec}(E_{Agg}, d, T, \alpha)\rightarrow \widehat{\nabla}$. After that, the server calls Algorithm \ref{alg:alg1} to extract aggregated gradients $\nabla$ from $\widehat{\nabla}$. Finally, the server updates $W$ with Eq. \eqref{eq:aggre}\;
        }
        }
	% \end{algorithmic}
	% \label{alg1}
\end{algorithm}

\subsection{System initialization}\label{subsec:sysini}
In our scheme, the server and all clients collaborate to complete the system initialization, which primarily involves the following two steps:

\begin{enumerate}
    \item \emph{System parameter generation}: $\mathbf{SysGen}(\kappa_{1}, \kappa_{2})\rightarrow (p, q, \mathbb{G}, g, \beta, \mathbf{a}, H(\cdot))$. The server is accountable for generating the system parameters, as illustrated below.
    \begin{itemize}
        \item Take a security parameter $\kappa_{1}$ as input, output the parameters $(p, q, \mathbb{G}, g)$, where $p$ is a uniformly chosen large prime number such that $|p|=\kappa_{1}$ and $p-1=\gamma q$ for a specified small integer $\gamma$, 
        $g$ is a uniformly chosen generator of the subgroup $\mathbb{G}$ of prime order $q$ of $\mathbb{Z}^{*}_{p}$. 
    
        \item Randomly select $\alpha \in \mathbb{Z}_{q}$ and set $\beta=g^{\alpha}$. 
        
        \item Choose a  cryptographic hash function $H: \mathbb{G} \rightarrow \{0,1\}^{\kappa_{2}}$, where $\kappa_{2}$ is the bit length of the symmetric key.
        
        \item Choose a super-increasing vector $\textbf{a}=(a_1=1,a_2,...,a_n)$, where $a_2,...,a_n$ are integers satisfying $\sum\nolimits_{j=1}^{i-1} a_jN\nabla_{max}<a_i$, $\sum\nolimits_{i=1}^{n} a_jN\nabla_{max}<p$, where $n$ is the dimension of gradient, $N$ denotes the number of clients and $\nabla_{max}$ is the maximum value of gradient. 
        It's worth noting that this sequence facilitates the compression of multidimensional data into 1-D, which is subsequently encrypted. As a result, the corresponding communication and computation overheads are significantly reduced. For additional information, please refer to \cite{YangLSTY19},\cite{LuLLLS12}.  

        \item Send the system parameter $\mathrm{SysPara}=(p, q, \mathbb{G}, g, \beta, \mathbf{a})$ to clients, and keep the private key $sk=\alpha$ secret.
    \end{itemize}
    \item \emph{Key generation:} $\mathbf{KeyGen}(p, q, \mathbb{G}, g, \beta, \mathbf{a})\rightarrow (pk_{i}, sk_{i})$.
    Given $(p, q, \mathbb{G}, g, \beta, \mathbf{a})$, each client $\mathcal{C}_{i}$ ($i\in \{1,2, \ldots, N\}$) randomly selects the private key $sk_{i} \in \mathbb{Z}_{q}$, and computes the corresponding public key $pk_i=g^{sk_i}$. After that, $\mathcal{C}_{i}$ sends $pk_i$ to the server for storage.
\end{enumerate} 
%Finally, the public system parameter is $(p, \mathbb{Z}^{*}_{p}, g, \beta, \mathbf{a})$, which is owned by all participants in the framework of federated learning.

\subsection{Model training and encryption}\label{subsec:training}
As described in Section \ref{subsec:system_model}, we cannot guarantee that all clients are capable of participating in every iteration. Thus, we should determine which clients can engage in the current iteration before executing model training.  A straightforward approach is for each client to transmit a 1-bit message to the server, signifying its online status. 
To facilitate the subsequent explanation, let $S$ denote the set of clients who are eligible to participate in model training in each iteration.

When determining the set $S$, the server performs the following operations: 
\begin{enumerate}
    \item Compute an aggregated public key $pk_{S}$ using the stored public key $pk_{i}$ of the client $\mathcal{C}_{i}\in S$. The corresponding aggregated public key generation function $\mathbf{AggPKGen}(pk_{i}, S)\rightarrow pk_{S}$ is defined as:
    \begin{equation*}
        pk_{S}=\prod\limits_{\mathcal{C}_i\in S}pk_i=\prod\limits_{\mathcal{C}_i\in S}g^{sk_i} =g^{\sum_{\mathcal{C}_i\in S}sk_i}.
    \end{equation*}
    \item For the client $\mathcal{C}_{i}\in S$, generate the corresponding symmetric key $key_{i}=H(pk_{i}^{\alpha})=H(g^{\alpha sk_{i}})$, where $H(\cdot)$ is the hash function such that $H: \mathbb{G}\rightarrow\{0,1\}^{\kappa_{1}}$ and $\kappa_{1}$ is the bit length of the symmetric key. Then, the server encrypts the current model parameter $W$ by the symmetric encryption algorithm (e.g., AES) as $W_{i}^{*}=SE_{key_{i}}(W)$. The corresponding symmetric encryption function is defined as $\mathbf{SE}(pk_{i}, \alpha, H(\cdot), W)\rightarrow W_{i}^{*}$.
    \item Send $(W_{i}^{*}, pk_{S})$ to each client $\mathcal{C}_{i}$ in the set $S$.
\end{enumerate}

%Thus, the server would not wait for the gradients of all clients to be collected before executing the above operations. That is, once a sufficient number of gradients are collected, the aggregation operation can be performed. For the convenience of subsequent description, let  $S$ denote the set of clients who successfully upload encrypted gradients.

Once receiving $(W_{i}^{*}, pk_{S})$, each client $\mathcal{C}_{i} \in S$ first checks whether $pk_{S}=pk_{i}$ holds. If holds, it indicates that only $\mathcal{C}_{i}$ is online at present. Under these circumstances, it is unnecessary for $\mathcal{C}_{i}$ to persist in uploading any data, as there is no merit for $\mathcal{C}_{i}$.
Otherwise, $\mathcal{C}_{i}$ calculates $key_{i}=H(\beta^{sk_{i}})$ and obtains $W$ by decrypting the symmetric ciphertext $W_{i}^{*}$ with $key_{i}$. The corresponding symmetric decryption operation is defined as $\mathbf{DSE}(\beta, H(\cdot), sk_{i}, W_{i}^{*})\rightarrow W$.

After that, $\mathcal{C}_{i}$ trains the model $W$ in several iterations with local dataset $\mathcal{D}_{i}$ to get the corresponding gradients $\nabla_{i}$. Note that for the deep neural network with $L$ layers, both $W$ and $\nabla_{i}$ actually consist of $L$ matrices, for example $W=\{W^{(l)}\in \mathbb{R}^{n_{l}\times n_{l-1}}\}^{L}_{l=1}$, where $n_{l}$ is the number of neural in the $l$-th layer. As we know, any matrix can be represented with a vector, so we can transform $W$ and $\nabla_{i}$ into the vectors, respectively, and the corresponding dimensions are $n$ (i.e., $W, \nabla_{i} \in \mathbb{R}^{n}$).

In order to protect the privacy of local training data, $\mathcal{C}_{i}$ will encrypt the local gradients $\nabla_{i}=(\nabla_{i1}, \nabla_{i2}, \ldots, \nabla_{in})$ before uploading them to the server. The details of the encryption function $\mathbf{Enc}(\nabla_{i}, pk_{S}, \mathrm{SysPara})\rightarrow E(\widehat{\nabla}_{i})$ are shown below. 
\begin{itemize}
    \item For the local gradient vector $\nabla_{i}$, compress it into a number (denoted as $\widehat{\nabla}_{i}$) via $\mathbf{a}$:    
    \begin{equation}\label{eq:compre}
        \widehat{\nabla}_{i}=a_{1}\nabla_{i1}+a_{2} \nabla_{i2}+\cdots+a_{n}\nabla_{in}<p, 
    \end{equation}
   \item Randomly choose $r_{i1}, r_{i2}\in \mathbb{Z}_{q}$, and compute the ciphertext (denoted as $E(\widehat{\nabla}_{i})=(E_{i1}, E_{i2})$) of $\widehat{\nabla}_{i}$ as: 
    \begin{equation*}
	\left\{\begin{aligned}
			E_{i1}&=(p+1)^{\widehat{\nabla}_{i}}\cdot \beta^{r_{i1}} \mod p^{2}\\
			E_{i2}&=(g^{r_{i2}} \mod{p^2},g^{r_{i1}}\cdot(pk_{S})^{r_{i2}} \mod{p^2})
	\end{aligned}
        \right.
    \end{equation*}
    Note that the random numbers $r_{i1}$ and $r_{i2}$ can differ for different iterations.
\end{itemize}
Finally, $\mathcal{C}_{i}$ sends $E(\widehat{\nabla}_{i})=(E_{i1}, E_{i2})$ to the server. 

\subsection{Model aggregation and decryption}\label{subsec:aggre}
Upon receiving the encrypted gradients from clients, the server initially aggregates them and subsequently decrypts the aggregated result via interaction with clients. Then, the server updates the current model $W$  for the next iteration. The following section expounds on the above procedures in detail.

\begin{enumerate}
    \item \emph{Challenge generation}: $\mathbf{ChelGen}(E(\widehat{\nabla}_{i}))\rightarrow R$.
    With the ciphertext $E_{i2}=(g^{r_{i2}} \mod{p^2},g^{r_{i1}}\cdot(pk_{S})^{r_{i2}} \mod{p^2})$, the server generates a challenge $R$ as:
$$R=\prod\limits_{\mathcal{C}_i\in S}g^{r_{i2}} \mod p^2=g^{\sum\nolimits_{\mathcal{C}_i \in S}r_{i2}} \mod p^2.$$
Then, the server sends $R$ to all clients in the set $S$. 

\item \emph{Response generation}: $\mathbf{ResGen}(R, E(\widehat{\nabla}_{i}))\rightarrow T_{i}$.
After receiving the challenge $R$, the client $\mathcal{C}_{i}\in S$ calculates a response, denoted as $T_{i}$:
$$T_i=R^{sk_i}=(g^{\sum\nolimits_{\mathcal{C}_i \in S}r_{i2}})^{sk_i} \mod{p^2}.$$
Then, $\mathcal{C}_{i}$ sends $T_i$ to the server.

\item After receiving the responses of all clients in $S$, the server first performs the aggregation operation. The corresponding aggregation function $\mathbf{Agg}(E(\widehat{\nabla}_{i}), T_{i})\rightarrow (E_{Agg}, d, T)$ is defined as follows: 
\begin{equation*}
    \left\{\begin{aligned}
        E_{Agg} & =  \prod\limits_{\mathcal{C}_i\in S} (p+1)^{\widehat{\nabla}_{i}} \beta^{r_{i1}} \mod p^{2} \\ 
        & =   (p+1)^{\sum\nolimits_{\mathcal{C}_i\in S}\widehat{\nabla}_{i}}\cdot g^{\alpha{\sum\nolimits_{\mathcal{C}_i\in S}r_{i1}}} \mod{p^{2}}, \\ \\
        d&=\prod\limits_{\mathcal{C}_i\in S}g^{r_{i1}}\cdot(pk_{S})^{r_{i2}} \mod{p^2} \\ 
        &=g^{\sum_{\mathcal{C}_i\in S}r_{i1}}\cdot (pk_{S})^{\sum_{\mathcal{C}_i\in S}r_{i2}} \mod{p^2},\\ \\
        T&=\prod\limits_{\mathcal{C}_i\in S}T_i= \prod\limits_{\mathcal{C}_i\in S}R^{sk_i} \mod{p^2}\\
        &=(g^{\sum\nolimits_{\mathcal{C}_i\in S}r_{i2}})^{\sum\nolimits_{\mathcal{C}_i\in S}sk_i} \mod{p^2}.
    \end{aligned}
    \right.
\end{equation*}
Then, the server performs the decryption operations with its own private key $\alpha$. The corresponding decryption function $\mathbf{Dec}(E_{Agg}, d, T, \alpha)\rightarrow \widehat{\nabla}$ is shown below: 
\begin{equation*}
    \left\{
        \begin{aligned}
        \frac{d}{T}&=g^{\sum\nolimits_{\mathcal{C}_i\in S}r_{i1}} \mod{p^2}, \\
        \widehat{\nabla}&=L\left(E_{Agg}\cdot \left(\frac{d}{T}\right)^{-\alpha} \mod{p^{2}}\right)=\sum\nolimits_{\mathcal{C}_i\in S}\widehat{\nabla}_{i},
        \end{aligned}
    \right.
\end{equation*}
where $L(x)=\frac{x-1}{p}$.
According to Eq. \eqref{eq:compre}, $\widehat{\nabla}=\sum\nolimits_{\mathcal{C}_i\in S}\widehat{\nabla}_{i}$ can be represented as:
\begin{equation*}   \widehat{\nabla}=a_1\sum\limits_{\mathcal{C}_i\in S}\nabla_{i1}+a_2\sum\limits_{\mathcal{C}_i\in S}\nabla_{i2}+\cdots+a_n\sum\limits_{\mathcal{C}_i\in S}\nabla_{in}.
\end{equation*}
By invoking the Algorithm \ref{alg:alg1}, the server can recover the aggregated data $\nabla=(\nabla_{1}, \nabla_{2}, \ldots,\nabla_{n})$, where for $j\in \{1,2,\ldots, n\}$, $\nabla_{j}=\sum\nolimits_{\mathcal{C}_i\in S}\nabla_{ij}$.
\item Finally, according to Eq. \eqref{eq:aggre}, the server performs the update operation with the aggregated data $\nabla=(\nabla_{1}, \nabla_{2}, \ldots,\nabla_{n})$ to obtain the updated global model parameter $W$. 
\end{enumerate}

\begin{algorithm}[htbp]
	\caption{Recover the Aggregated Gradient}\label{alg:alg1}
		\KwData
        {$\nabla=a_{1}\nabla_{1}+a_{2}\nabla_{2}+\cdots+a_{n}\nabla_{n}$ and a super-increasing sequence $\textbf{a}=(a_{1}=1,\cdots ,a_n)$ with $\sum\nolimits_{j=1}^{i-1} a_j\nabla_{j}<a_i$, $\sum\nolimits_{i=1}^{n} a_j\nabla_{j}<p$, for $i=2,\cdots,n$.}
		\KwOut{$(\nabla_{1},\nabla_{2},\cdots,\nabla_{n}).$}
		Set $t_n=\nabla$\;
		\For{$i=n$ to $2$}{
    		$t_{i-1} = t_i \, mod \, a_i$\;
    		$\nabla_{i} = \frac{t_i-t_{i-1}}{a_i}$\;
		}
		$\nabla_1=t_1$\;
		\textbf{return} $(\nabla_1,\nabla_2,\cdots,\nabla_n)$
	\label{alg1}
\end{algorithm}

\begin{remark}[The correctness of recovering $\nabla$]\label{remark}
Because $\nabla_{max}$ is the maximum value of gradient, i.e., $\nabla_{max}\geq\forall_{i,j}\nabla_{ij}$, where $i\in \{1, 2, \ldots, N\}$ and $j\in \{1,2,\ldots, n\}$, we have $\nabla_{j}=\sum_{i\in S}\nabla_{ij}\le S\nabla_{max} \le N\nabla_{max}$, which indicates 
\begin{equation*}  
    \begin{aligned}
    \widehat{\nabla}&=a_1\nabla_{1}+a_2\nabla_{i}+\cdots+a_n\nabla_{n} \\
    &\le a_1N\nabla_{max}+a_2N\nabla_{max}+\cdots+a_nN\nabla_{max} \\
    &=\sum\nolimits_{i=1}^{n} a_jN\nabla_{max}.
    \end{aligned}
\end{equation*}
As defined in Section \ref{subsec:sysini}, $\sum\nolimits_{i=1}^{n} a_j{N}\nabla_{max}<p$, so we have $\widehat{\nabla}<p$. That is, the data $\widehat{\nabla}$ meets the message space of the encryption algorithm and can be correctly decrypted by running $\mathbf{Dec}(E_{Agg}, d, T, \alpha)$. 

Next, we show how to obtain $(\nabla_{1}, \nabla_{2}, \ldots,\nabla_{n})$ from $\widehat{\nabla}$ by running algorithm \ref{alg:alg1}. More specifically, in algorithm \ref{alg:alg1}, given the $\widehat{\nabla}$ as input, we first set $t_n=\widehat{\nabla}$. Since $\sum\nolimits_{j=1}^{i-1} a_j\nabla_{j}\le\sum\nolimits_{j=1}^{i-1} a_jN\nabla_{max}<a_i$, we have 
\begin{equation*}
    \begin{aligned}    a_{1}\nabla_{1}+a_{2}\nabla_{2}+\cdots+a_{n-1}\nabla_{n-1} < a_n,
    \end{aligned}
\end{equation*}
Therefore, $t_{n-1} = t_{n} \mod a_{n} = a_{1}\nabla_{1} +a_{2}\nabla_{2}\cdots + a_{n-1}\nabla_{n-1}$, and we can obtain the gradient
\begin{equation*}
    \begin{aligned}    
    \frac{t_{n}-t_{n-1}}{a_{n}}=\frac{a_{n}\nabla_{n}}{a_{n}}=\nabla_{n}
    \end{aligned}
\end{equation*}
With the similar procedure, we can also prove each $\nabla_{j}=\sum_{i\in S}\nabla_{ij}$ for $j=1,2,\ldots, n-1$. As a result, we prove the server can obtain the correct aggregated local gradient $\nabla=(\nabla_{1}, \nabla_{2}, \ldots,\nabla_{n})$.
\end{remark}

\subsection{Extension}
As we are all aware, the dimension of the deep model parameter is extensive. Thus, it is almost insurmountable to compress such high-dimensional data into one at a time, owing to the message space constraint in the encryption algorithm. Therefore, in this section, we extend our scheme to preclude overflows in ciphertext additions.

More specifically, if the merged data $\widehat{\nabla}_{i} > p$, then we can split $\nabla_{i}=(\nabla_{i1}, \nabla_{i2}, \ldots, \nabla_{in})$ into multiple small segments and compress each segment into the 1-D data by the super-increasing vector $\mathbf{a}$, denoted as $\widehat{\nabla}_{i}=(\widehat{\nabla}_{i1}, \widehat{\nabla}_{i2}, \cdots, \widehat{\nabla}_{iu})$, where $\widehat{\nabla}_{ij}=a_{1}\nabla_{i,jk-k+1}+a_{2} \nabla_{i,jk-k+2}+\cdots+a_{k}\nabla_{i,jk}<p$ for $j\in [1, u]$ and $u=\lceil\frac{n}{k}\rceil$ is the number of split segments. To this end, the difference between the basic version and the extension is that in the system initialization, the server needs to generate $u$ pairs of private and public keys, denoted as $\vec{\alpha} = (\alpha_{1}, \alpha_{2}, \cdots, \alpha_{u})$ and $\vec{\beta} = (\beta_{1}, \beta_{2}, \cdots, \beta_{u})$ where $\beta_{i}=g^{\alpha_{i}}$, and then sends the system parameter $\mathrm{SysPara}=(p, q, \mathbb{G}, g, \vec{\beta}, \mathbf{a})$ to clients. 

After computing the local gradient $\nabla_{i}$, each client $\mathcal{C}_{i}$ compresses $\nabla_{i}$ into $\widehat{\nabla}_{i}$, and then conducts the encryption function $\mathbf{Enc}(\nabla_{i}, pk_{S}, \mathrm{SysPara})$ to obtain $E(\widehat{\nabla}_{i})=(E_{i1}, E_{i2})$, where $E_{i1}=(E(\widehat{\nabla}_{i1}), E(\widehat{\nabla}_{i2}), \ldots, E(\widehat{\nabla}_{iu}))$ and $E_{i2}$ are generated as:  
  \begin{equation*}
	\left\{\begin{aligned}
			E(\widehat{\nabla}_{ij})&=(p+1)^{\widehat{\nabla}_{ij}}\cdot \beta_{j}^{r_{i1}} \mod p^{2}, ~~for~~j\in [1, u]\\
			E_{i2}&=(g^{r_{i2}} \mod{p^2},g^{r_{i1}}\cdot(pk_{S})^{r_{i2}} \mod{p^2})
	\end{aligned}
        \right.
    \end{equation*}
Then, the remaining steps are almost the same as the Section \ref{subsec:aggre} in the basic version. Obviously, in this case, the number of ciphertexts is $\mathbf{u+2}$. It is worth noting that we can directly use the original variant ElGamal to encrypt $\widehat{\nabla}_{i}=(\widehat{\nabla}_{i1}, \widehat{\nabla}_{i2}, \cdots , \widehat{\nabla}_{iu})$ as $(g^{r_{ij}},g^{\widehat{\nabla}_{ij}}\cdot(pk_{S})^{r_{ij}})$ for $j\in [1, u]$, which can still be decrypted correctly by our decryption function. Obviously, the number of ciphertexts generated in this way is $\mathbf{2u}$, which is larger than $\mathbf{u+2}$ when $u>2$. Therefore, our skillful modification offers two distinct advantages: 
\begin{itemize}
    \item Exclusively the server can decrypt the aggregated result, which is more secure than secret sharing-based schemes that overlook safeguarding the aggregated result.
    \item For deep models with large-scale parameters (i.e., $u$ is relatively large), computational costs and communication overhead are significantly diminished, as compared to the naive encryption algorithm.
\end{itemize}

\section{Security and fault tolerance analysis}\label{sec:security}
This section commences by analyzing the security properties of our scheme. Subsequently, we demonstrate the robustness in tolerating dropped clients.
     
\subsection{Security analysis}\label{subsec:security}
In particular, adhering to the design goals elucidated in Section \ref{subsec:design}}, our security analysis encompasses three aspects: 1) The privacy-preservation of individual local gradient, 2) The privacy-preservation of the aggregated result, and 3) The robustness against collusion.

\subsubsection{Privacy-preservation of individual local gradient}
This segment concentrates on how our proposed scheme can thwart both the internal adversary (the server and curious clients) and the external eavesdropper from acquiring the local gradient of a specific client. 
Before delving into the specifics, we would like to provide some relevant complexity assumptions \cite{YangLSTG21}, which substantiate the security proof of our scheme.

\begin{definition}[Discrete Logarithm (DL) Problem]\label{def:DL}
The DL problem in $\mathbb{G}$ is stated as follows: given $\beta \in \mathbb{G}$, compute $\alpha \in \mathbb{Z}_{q}$ such that $\beta=g^{\alpha}$.
\end{definition}

\begin{definition}[Computational Diffie-Hellman (CDH) Problem]\label{def:CDH}
The CDH problem in $\mathbb{G}$ is stated as follows: Given $g, g^{a}, g^{b} \in  \mathbb{G}$ for unknown $a, b \in \mathbb{Z}_{q}$, compute $g^{ab}$. 
\end{definition}

\begin{definition}[Decisional Diffie–Hellman (DDH) assumption]\label{def:ddh}
 The DDH assumption in $\mathbb{G}$ is stated as: given $(g, g^{a}, g^{b}, \omega)$ for $g\in \mathbb{G}$, $\omega \in \mathbb{G}$ and unknown $a, b \in \mathbb{Z}_{q}$, no probabilistic, polynomial-time algorithm $\mathcal{B}$ can determine whether $\omega=g^{ab}$ or a random element from $\mathbb{G}$ with more than a negligible function $\mathrm{negl(\kappa)}$, i.e., 
\begin{align*}
    \mathrm{DDH-Adv}_{\mathcal{B}}=&|\Pr[\mathcal{B}(g, g^{a}, g^{b}, g^{ab} )=1] \\
    &- \Pr[\mathcal{B}(g, g^{a}, g^{b}, \omega)=1]| \leq \mathrm{negl(\kappa)}.
\end{align*}
\end{definition}

In what follows, we show the details of the semantic security of our proposed scheme. Without loss of generality, we suppose the adversary $\mathcal{A}$ (maybe an external eavesdropper, the server, or the client $\mathcal{C}_{j}$) tries to obtain the client $\mathcal{C}_{i}$'s local gradient $\nabla_{i}$. As described in Section \ref{sec:scheme}, the client $\mathcal{C}_{i}$ needs to uploads the public key $Pk_{i}=g^{sk_{i}}$, the ciphertext $E(\widehat{\nabla}_{i})$ and the response $T_{i}=R^{sk_{i}}$. Obviously, given $g$, $R$, $Pk_{i}=g^{sk_{i}}$ and $T_{i}=R^{sk_{i}}$, $\mathcal{A}$ cannot obtain the private key $sk_{i}$ from $Pk_{i}$ and $T_{i}$ due to the difficulty of solving the DL problem (see Definition \ref{def:DL}). Therefore, we focus on analyzing the security of the ciphertext $E(\widehat{\nabla}_{i})$, which is encrypted as:
\begin{equation*}
	\left\{\begin{aligned}
			E_{i1}&=(p+1)^{\widehat{\nabla}_{i}}\cdot \beta^{r_{i1}} \mod p^{2},\\
			E_{i2}&=(g^{r_{i2}} \mod{p^2},g^{r_{i1}}\cdot(pk_{S})^{r_{i2}} \mod{p^2}),
	\end{aligned}
        \right.
    \end{equation*}
    where $r_{i1}$ and $r_{i2}$ are randomly chosen from $\mathbb{Z}_{q}$.

If an adversary $\mathcal{A}$ wants to obtain $(p+1)^{\widehat{\nabla}_{i}}$, then $\mathcal{A}$ needs to know $\beta^{r_{i1}}$. Note that $\beta=g^{\alpha}$ and $\alpha$ is selected by the server, so if $\mathcal{A}$ (corrupting the server) obtains $g^{r_{i1}}$, then $\beta^{r_{i1}}$ can be calculated as $\beta^{r_{i1}}=(g^{r_{i1}})^{\alpha}=(g^{\alpha})^{r_{i1}}$. As a result, the core of the security in our scheme comes down to whether $\mathcal{A}$ can obtain $g^{r_{i1}}$. Obviously, $g^{r_{i1}}$ can be regarded as the plaintext $m$, which is encrypted in the ciphertext $E_{i2}$. Therefore, in the following parts, we focus on the security analysis of the ciphertext $E_{i2}$. Since $pk_{S}=g^{\sum_{\mathcal{C}_i\in S}sk_i}$, we directly let $sk_{S}=\sum_{\mathcal{C}_i\in S}sk_i$, then $E_{i2}$ can be simplified as $E_{i2}=(g^{r_{i2}},m\cdot g^{sk_{S}\cdot r_{i2}})$, where $m=g^{r_{i1}}$. The corresponding security  is given in Theorem \ref{the:security}.

\begin{theorem}\label{the:security}
    Our proposed scheme is semantically secure against the chosen-plaintext attack under the DDH assumption.
\end{theorem}
\begin{proof}\label{proof:theo1}
Suppose a polynomial-time adversary $\mathcal{A}$ (maybe the server, a client $\mathcal{C}_{j}$ or an external eavesdropper)  can attack our scheme with advantage $\epsilon(\kappa)$, then we can build an algorithm $\mathcal{B}$ that can break the DDH problem with the same advantage as follows.
\begin{itemize}
     \item \emph{Init:} Given the parameters $(p,q, \mathbb{G}, g, g^{r}, g^{sk_{S}}, \omega)$\footnote{Note that the super-increasing vector $\mathbf{a}$ is used for compressing a vector into the data to reduce the computational and communication overheads. The corresponding operation is performed in the plaintext domain, so it has nothing to do with security, and can be ignored here. In addition, the parameter $\beta$ is not included in the ciphertext $E_{i2}$, and thus we ignore it here.}, $\mathcal{B}$ sets $pk_{S}= g^{sk_{S}}$, and gives the public parameters $(p,q, \mathbb{G}, g, pk_{S})$ to $\mathcal{A}$. 
         
     \item \emph{Challenge:} $\mathcal{A}$ selects two messages $m_{0}, m_{1}\in \mathbb{G}$, and then submits them to $\mathcal{B}$. $\mathcal{B}$ flips a fair binary coin $b$, and returns an encryption of $m_{b}\in \{m_{0}, m_{1}\}$. The ciphertext is output as 
      $$E_{b}=(g^{r}, m_{b}\cdot\omega).$$
      
     \item \emph{Guess:} $\mathcal{A}$ outputs a guess $b^{*}$ of $b$. If $b^{*}=b$, $\mathcal{B}$ outputs $1$ to indicate that $\omega=g^{r\cdot sk_{S}}$; Otherwise, $\mathcal{B}$ outputs $0$ to indicate that $\omega$ is a random element from $\mathbb{G}$ (i.e., $w=g^{x}$ for random $x\in \mathbb{Z}_{q}$).
 \end{itemize}

When $\omega=g^{r\cdot sk_{S}}$, then $\mathcal{A}$ sees a proper encryption of $m_{b}$, i.e., $E_{b}=(g^{r}, m_{b}g^{r\cdot sk_{S}})=(g^{r}, m_{b}(pk_{S})^{r})$. The advantage of $\mathcal{A}$ is $\epsilon(\kappa)$ by definition, i.e., $\mathcal{A}$ can obtain $m_{b}$ with advantage $\epsilon(\kappa)$ from the ciphertext $E_{b}$, so we have $\Pr[\mathcal{A}(b^{*}=b)]=\frac{1}{2}+\epsilon(\kappa)$. Since $\mathcal{B}$ outputs $1$ exactly when the output $b^{*}$ of $\mathcal{A}$ is equal to $b$, we have
\begin{equation*}
    \Pr\left[\mathcal{B}\left(g, g^{r}, g^{sk_{S}}, g^{r\cdot sk_{S}} \right)=1\right]=\Pr[\mathcal{A}(b^{*}=b)]=\frac{1}{2}+\epsilon(\kappa).
\end{equation*}
When $w=g^{x}$ is a random element from $\mathbb{G}$, then $E_{b}=(g^{r}, m_{b}g^{x})$ is not an actual encryption scheme (i.e., $m_{b}g^{x}$ is a random element of $\mathbb{G}$ from $\mathcal{A}$'s view), which means that $\mathcal{A}$ gains no information about $b$ except blinding guess. Hence, $\Pr[\mathcal{A}(b^{*}=b)]=\frac{1}{2}$, which implies that 
\begin{equation*}
     \Pr\left[\mathcal{B}\left(g, g^{r}, g^{sk_{S}}, g^{x} \right)=1\right]=\Pr[\mathcal{A}(b^{*}=b)]=\frac{1}{2}.
\end{equation*}
Therefore, we can obtain that 
\begin{align*}
    \mathbf{DDH-Adv}_{\mathcal{B}}=&|\Pr[\mathcal{B}(g, g^{r}, g^{sk_{S}}, g^{r\cdot sk_{S}} )=1] \\
    &- \Pr[\mathcal{B}(g, g^{r}, g^{sk_{S}}, g^{x} )=1]| \\
    =& \left|\frac{1}{2}+\epsilon(\kappa)-\frac{1}{2}\right|=\epsilon(\kappa),
\end{align*}
which implies that $\epsilon(\kappa) \leq \mathrm{negl(\kappa)}$ with Definition \ref{def:ddh}.
\end{proof}

From the Theorem \ref{the:security}, we can know that $\mathcal{A}$ cannot obtain $g^{r_{i1}}$ from the ciphertext $E_{i2}$, which implies that $\mathcal{A}$ cannot compute $\beta^{r_{i1}}$ even if the server knows $\alpha$. It is worth noting that neither external eavesdroppers nor clients know about $\alpha$, so it is even less likely for them to get $\beta^{r_{i1}}$ compared with the server. Therefore, our scheme can ensure any adversary including the internal participants (i.e., the server or the client $\mathcal{C}_{j}$ for $j\neq i$ ) and the external eavesdropper cannot get $\widehat{\nabla}_{i}$ of a particular client $\mathcal{C}_{i}$.

\subsubsection{Privacy-preservation of aggregated result}
To our knowledge, many state-of-the-art secure aggregated works  \cite{BonawitzIKMMPRS17},\cite{SoNYL0AGA22} do not take into account the privacy-preservation of the aggregated result. In other words, any adversary other than the server, such as eavesdroppers or curious clients, can obtain the aggregated results, as long as they have access to the uploaded data. Contrarily, our designed scheme in this paper circumvents this situation by ensuring that only the server can decrypt the aggregated encrypted gradient.

\begin{theorem}\label{the:aggregation}
   The aggregated gradient is semantically secure against the chosen-plaintext attack launched by curious clients or eavesdroppers under the DDH assumption.
\end{theorem}
\begin{proof}
    Suppose an adversary $\mathcal{A}$ (e.g., eavesdroppers or curious clients) can obtain the data communicated between each client and the server. In this case, $\mathcal{A}$ can obtain $pk_{S}$, $R$ and $\{E(\widehat{\nabla}_{i}), T_{i}\}_{\mathcal{C}_{i}\in S}$. Obviously, $\mathcal{A}$ can perform the aggregation function $\mathbf{Agg}(E(\widehat{\nabla}_{i}), T_{i})\rightarrow (E_{Agg}, d, T)$, and compute $d/T$ to get $g^{\sum_{\mathcal{C}_{i}\in S}r_{i1}}$. Thus, the crux is that given  $(p+1)^{\sum\nolimits_{\mathcal{C}_i\in S}\widehat{\nabla}_{i}}\cdot g^{\alpha{\sum\nolimits_{\mathcal{C}_i\in S}r_{i1}}}$ and $g^{\sum_{\mathcal{C}_{i}\in S}r_{i1}}$, whether $\mathcal{A}$ can recover $(p+1)^{\sum\nolimits_{\mathcal{C}_i\in S}\widehat{\nabla}_{i}}$. Note that $\beta=g^{\alpha}$ is the public key, which can be obtained by $\mathcal{A}$. Hence, the problem comes down to given $g^{\alpha}$ and $g^{\sum_{\mathcal{C}_{i}\in S}r_{i1}}$, whether $\mathcal{A}$ can decrypt $(p+1)^{\sum\nolimits_{\mathcal{C}_i\in S}\widehat{\nabla}_{i}}\cdot g^{\alpha{\sum\nolimits_{\mathcal{C}_i\in S}r_{i1}}}$ or not. Obviously, this problem can be attributed to the DDH problem. Therefore, similar to the proof of Theorem \ref{the:security}, the value $(p+1)^{\sum\nolimits_{\mathcal{C}_i\in S}\widehat{\nabla}_{i}}\cdot g^{\alpha{\sum\nolimits_{\mathcal{C}_i\in S}r_{i1}}}$ meets the semantically secure under the DDH assumption. The details of the proof are omitted due to the page limitation. 
\end{proof}

It is noteworthy that our scheme offers the advantage of ensuring that when the model is well-trained, only the server can access the well-trained model and provide the prediction service \footnote{During the model training, the server sends the encrypted model $W^{*}_{i}$ to each client $\mathcal{C}_{i}\in S$ and only the client in $S$ can decrypt the symmetric ciphertext to obtain $W$ (i.e., other offline clients and eavesdroppers cannot get $W$). Although these clients can obtain the global model $W$, this model is the result of an intermediate process rather than the final well-trained model. Thus, its leakage will not affect the final result too much. }.  This scenario is quite common in real-world applications.

 For example, when the server is a service provider and clients are data providers, the server wants to offer prediction services based on a well-trained deep model. Evidently, this well-trained deep model is a valuable asset, which must be safeguarded from disclosure. 
 However, the prerequisite is that the server must compensate the data provider to generate this model. It is apparent that data providers will not allow their training data to be leaked, as training data is their property and involves their privacy.
 For example, the training architecture of the privileged party is introduced in \cite{SongWWTLRWH22},  where the privileged party dominates the model training, and only the privileged party can recover the final model. Clearly, in our scheme, the server is akin to the privileged party.

\subsubsection{Robustness against collusion}
As stipulated in the Threat Model, the server may collude with some clients to jointly obtain the local gradients of other clients, thereby enhancing the attack capability. Thus, in this section, we demonstrate how our scheme can thwart collusion between the server and curious clients.

\begin{theorem}\label{the:collusion}
    Our scheme guarantees the privacy-preservation of the individual local gradient as long as any two online clients do not collude, i.e., our scheme can withstand collusion between the server and up to any $N-2$ clients.  
\end{theorem}
\begin{proof} 
    Without loss of generality, let's assume that two online clients $\mathcal{C}_{i}$ and $\mathcal{C}_{j}$ do not collude, whereas the remaining $N-2$ clients collude with the server to try to obtain the local gradients of $\mathcal{C}_{i}$ and $\mathcal{C}_{j}$. As proved in Theorem \ref{the:security}, neither the server nor the client can obtain $\widehat{\nabla}_{i}$ or $\widehat{\nabla}_{j}$ from $E(\widehat{\nabla}_{i})$ or $E(\widehat{\nabla}_{j})$, respectively. Thus, they have to address it from other transmitted data. Note that since $N-|S|$ dropped clients would not join the model training, they are clearly unlikely to breach the privacy guarantee, regardless of whether they collude or not. Therefore, we only consider the case where online $|S|-2$ curious clients (i.e., $S/\{\mathcal{C}_{i}, \mathcal{C}_{j}\}$) collude with the server for simplicity. In this scenario, these clients can collaborate to compute $\sum_{\mathcal{C}_k\in S/\{\mathcal{C}_{i}, \mathcal{C}_{j}\}}sk_{k}$ and $\sum_{\mathcal{C}_k\in S/\{\mathcal{C}_{i}, \mathcal{C}_{j}\}}r_{k2}$\footnote{The simple way is that all colluded clients send their parameters to one client (such as $\mathcal{C}_{k}$) for computing. It is worth noting that even though $\mathcal{C}_{k}$ knows other client's secret parameters, for example, $\mathcal{C}_{l}$'s $sk_{l}$ and $r_{l2}$, it still cannot obtain $(p+1)^{\widehat{\nabla}_{l}}$. Specifically, given $E(\widehat{\nabla}_{l})$, $r_{l2}$ and $sk_{l2}$, $\mathcal{C}_{k}$ can obtain $g^{r_{l1}}$ by computing $g^{r_{l1}}(pk_{S}^{r_{l2}})^{-r_{l2}}$. After that, given $\beta=g^{\alpha}$ and $g^{r_{l1}}$, it cannot obtain $\beta^{r_{l1}}$ due to the difficulty of solving the CDH problem (see Definition \ref{def:CDH}). Thereby, $\mathcal{C}_{k}$ cannot get $(p+1)^{\widehat{\nabla}_{l}}$ from $E_{l1}$.}, which are then sent to the server. However, although the colluded client $\mathcal{C}_k$ wants to obtain $\nabla_{i}$ of a particular client $\mathcal{C}_{i}$, its own local gradient $\nabla_{k}$ should remain protected from others. Hence, $\mathcal{C}_k$ does not directly sends $sk_{k}$ and $r_{k2}$ to the server \footnote{Obviously, if the server knows $sk_{k}$ and $r_{k2}$, then it can compute $pk_{S}^{r_{k2}}$. Thereby it can directly obtain $(p+1)^{\widehat{\nabla}_{i}}$ by computing $g^{r_{k1}}pk_{S}^{r_{k2}}/pk_{S}^{r_{k2}}=g^{r_{k1}}$ and $E_{k1}/(g^{r_{k1}})^{\alpha}=(p+1)^{\widehat{\nabla}_{k}}$.}. In this situation, besides the normal execution of the protocol, the server calculates the following values without the awareness of $\mathcal{C}_{i}$ and $\mathcal{C}_{j}$. 
    \begin{equation*}
    \left\{\begin{aligned}
        E_{i\& j} & = (p+1)^{\widehat{\nabla}_{i}} \beta^{r_{i1}} \cdot  (p+1)^{\widehat{\nabla}_{j}} \beta^{r_{j1}}\\ 
        & =   (p+1)^{\widehat{\nabla}_{i}+\widehat{\nabla}_{j}}\cdot g^{\alpha(r_{i1}+r_{j1})} \mod{p^{2}}, \\ 
        d_{i\& j}&=g^{r_{i1}}\cdot(pk_{S})^{r_{i2}} \cdot g^{r_{j1}}\cdot(pk_{S})^{r_{j2}}
        \\ 
        &=g^{r_{i1}+r_{j1}}(pk_{S})^{r_{i2}+r_{j2}}=g^{r_{i1}+r_{j1}}(g^{\sum\nolimits_{\mathcal{C}_k\in S}sk_{k}})^{r_{i2}+r_{j2}}\\ 
        T_{i\& j}&=T_i\cdot T_j = R^{sk_i} \cdot R^{sk_j}=(g^{\sum\nolimits_{\mathcal{C}_k\in S}r_{k2}})^{sk_i+sk_{j}}. \\    
    \end{aligned}
    \right.
\end{equation*}
Obviously, $g^{r_{i1}+r_{j1}}$ cannot be deduced from $\frac{d_{i\& j}}{T_{i\& j}}$. Thus, the server also needs to  
calculate the following values using $\sum_{\mathcal{C}_k\in S/\{\mathcal{C}_{i}, \mathcal{C}_{j}\}}sk_{k}$ and $\sum_{\mathcal{C}_k\in S/\{\mathcal{C}_{i}, \mathcal{C}_{j}\}}r_{k2}$: 
\begin{equation*}
  \left\{\begin{aligned}
    d^{*}_{i\& j}&= d_{i\& j} \cdot (g^{r_{i2}}\cdot g^{r_{j2}})^{-\sum_{\mathcal{C}_k\in S/\{\mathcal{C}_{i}, \mathcal{C}_{j}\}}sk_{k}}\\
    &=g^{r_{i1}+r_{j1}}g^{(sk_{i}+sk_{j})(r_{i2}+r_{j2})} \\
    T^{*}_{i\& j} &= T_{i\& j} \cdot (pk_{i}\cdot pk_{j})^{-\sum_{\mathcal{C}_k\in S/\{\mathcal{C}_{i}, \mathcal{C}_{j}\}}r_{k2}}\\
    &=g^{(sk_{i}+sk_{j})(r_{i2}+r_{j2})}.
  \end{aligned} 
  \right.
\end{equation*}    
Thereby, the server can compute $d^{*}_{i\& j}/T^{*}_{i\& j}=g^{r_{i1}+r_{j1}}$. With the secret $\alpha$, the server calculates $E_{i\& j}/(g^{r_{i1}+r_{j1}})^{\alpha}=(p+1)^{\widehat{\nabla}_{i}+\widehat{\nabla}_{j}}$. Obviously, these colluded participants can only obtain the aggregated result $\widehat{\nabla}_{i}+\widehat{\nabla}_{j}$  rather than the individual gradient $\widehat{\nabla}_{i}$ or $\widehat{\nabla}_{i}$. Therefore, our scheme is resistant to collusion between the server and up to $N-2$ clients.  
\end{proof}

\subsection{Robustness in tolerating dropped clients}
In this part, we demonstrate the robustness of our scheme in tolerating dropped clients, as described in Theorem \ref{the:fault}.
\begin{theorem}\label{the:fault}
    Our proposed scheme is capable of achieving robustness in tolerating up to $N-2$ dropped clients.
\end{theorem}
\begin{proof}
     As outlined in Section \ref{subsec:training}, the server needs to initially identify the clients who can participate in the model training and decryption, i.e., determine the set $S$ of online clients. Subsequently, only the clients in $S$ (i.e., online clients) will perform the subsequent model training and aggregation. Therefore, based on the correctness of recovering the aggregated gradient $\nabla$ (see the Remark \ref{remark}), our proposed scheme is robust for the dropped clients. 
     Next, we explain why the up-bound number of dropped clients is $N-2$ instead of $N-1$. Consider the case where the number of clients in set 
      $S$ is $1$ (i.e., $N-1$ clients cannot participate in the current iteration of model training), without loss of generality, let $S=\{\mathcal{C}_{1}\}$, then the aggregated public key is $pk_{S}=pk_{1}=g^{sk_{1}}$.  In this scenario, only $\mathcal{C}_{1}$ uploads the ciphertext $E(\widehat{\nabla}_{1})=(E_{11}, E_{12})$. 
      The server then generates $R=g^{r_{12}}$ and forwards it to $\mathcal{C}_{1}$. following which $\mathcal{C}_{1}$ computes $T_{1}=R^{sk_{1}}=g^{r_{12}\cdot sk_{1}}$ and sends it to the server. In this case, the server can obtain the aggregated result as:    
     \begin{equation*}
    \left\{\begin{aligned}
        E_{Agg} & =E_{11}= (p+1)^{\widehat{\nabla}_{1}}\cdot g^{\alpha r_{11}} \mod{p^{2}} \\ 
        d&=g^{r_{11}}\cdot(pk_{S})^{r_{12}}=g^{r_{11}}\cdot g^{sk_{1} \cdot r_{12}} \\ 
        T&=T_{1}=g^{r_{12}\cdot sk_{1}}.
    \end{aligned}
    \right.
\end{equation*}
The server can obtain $\widehat{\nabla}_{1}$, which is the local gradient of $\mathcal{C}_{1}$, by computing $L(E_{Agg}\cdot (d/T)^{-\alpha} \mod{p^{2}})$ with the private key $\alpha$. In other words, if the number of clients in set $S$ is $1$, then the corresponding client's local gradient will be exposed. That's why we need each client $\mathcal{C}_{i}$ to check whether $pk_{S}=pk_{i}$ holds after receiving the aggregated public key $pk_{S}$ (see Section \ref{subsec:training}). If it holds, $\mathcal{C}_{1}$ will discontinue the following operations for privacy preservation. As a result, we need to limit the up-bound number of dropped clients to $N-2$ to preserve the privacy of the local gradient of a single client.  
\end{proof}

\begin{table*}[htb]
\begin{center}
\caption{Comparison of Security for secure Aggregation Schemes.}\label{tab:secure_compa}
\begin{threeparttable}
\begin{tabular}{c|ccccc}
\toprule
  & Scheme \cite{BonawitzIKMMPRS17} &Scheme \cite{SoNYL0AGA22} & Scheme \cite{PhongAHWM18} &Scheme \cite{FangGHMFY21}&Our scheme \\
  \midrule
  Without trusted third party participation &\checkmark &\checkmark & \ding{53} & \ding{53}& \checkmark\\
  Client-defined private key & \checkmark & \checkmark & \ding{53} & \checkmark & \checkmark\\
  Confidentiality of individual users & \checkmark & \checkmark&\checkmark &\checkmark &\checkmark \\
   Decryption right of the aggregated result & \ding{53}& \ding{53} &Clients & \ding{53}& Server\\
   Privacy guarantee against colluding clients & $t-1$ & $t-1$& \ding{53} & $t-1$& $N-2$ \\
   Robustness against dropped clients & $N-t$& $N-t$&$N-2$ & $N-t$& $N-2$\\
 \bottomrule
\end{tabular}
        \begin{tablenotes}
         \item As analyzed in \cite{BonawitzIKMMPRS17}, it is best to set $t \geq \lfloor \frac{2n}{3}+1\rfloor$ for the privacy consideration, and thus the schemes \cite{BonawitzIKMMPRS17, SoNYL0AGA22, FangGHMFY21}  can only tolerate a relatively small number of dropped clients.
       \end{tablenotes}
\end{threeparttable}
\end{center}
\end{table*}

\subsection{Comparison} 
In this section, we list a comparison of the privacy and robustness in tolerating dropped clients of the state-of-the-art secure aggregation schemes in Table \ref{tab:secure_compa}, where \checkmark and \ding{53} indicate satisfaction and dissatisfaction, respectively. Before discussing the comparison result, we briefly introduce these five schemes. Schemes \cite{BonawitzIKMMPRS17} and \cite{SoNYL0AGA22} are designed using the $(t, n)$-threshold secret sharing technique, while both schemes \cite{PhongAHWM18} and our scheme are based on the additive homomorphic encryption technique. Scheme \cite{FangGHMFY21} is introduced by combining these two techniques, where the $(t, n)$-threshold secret sharing technique and the additive homomorphic encryption technique are used for private key distribution and secure additive aggregation, respectively. 

First of all, since the primary goal of secure aggregation in federated learning is to safeguard the local gradient, all schemes naturally ensure the confidentiality of individual clients. 

For the two secret sharing-based schemes \cite{BonawitzIKMMPRS17} and \cite{SoNYL0AGA22}, the $(t, n)$-threshold secret sharing technique enables clients to select their own private keys and distribute the shares of the private key to other clients, allowing both schemes to satisfy the properties of no trusted third-party participation and client-defined private key. Additionally, as long as $t$ encrypted gradients are received, the aggregated result can be obtained. Thus, neither of the two schemes considers the decryption right of the aggregated result, meaning that anyone who obtains $t$ encrypted gradients can recover the aggregated result. However, as we know, $(t, n)$-threshold secret sharing-based schemes must balance the trade-off between privacy guarantee against colluding clients (i.e., up to $t-1$ colluding clients) and robustness against dropped clients (i.e., up to $N-t$ clients). That is, these schemes increase the privacy guarantee by reducing the robustness against dropped clients and vice versa.

For two homomorphic encryption-based schemes \cite{PhongAHWM18} and \cite{FangGHMFY21}, scheme \cite{PhongAHWM18} requires the selection of a private key $sk$ that is shared among all clients. The simplest way to accomplish this is to introduce a trusted third party. Similar to scheme \cite{PhongAHWM18}, scheme \cite{FangGHMFY21} selects a private key $sk$ during system initialization and computes the corresponding $N$ shares $\{sk_{1}, sk_{2}, \ldots, sk_{N}\}$ using the secret sharing technique, and then distribute $sk_{i}$ to the client $\mathcal{C}_{i}$. 
To ensure privacy, a trusted third party must be introduced to complete this operation. As a result, both schemes cannot satisfy the property of no trusted third-party participation. Since all clients share a private key, scheme \cite{PhongAHWM18} cannot satisfy the property of the client-defined private key. Additionally, similar to \cite{BonawitzIKMMPRS17} and \cite{SoNYL0AGA22}, scheme \cite{FangGHMFY21} cannot satisfy the decryption right of the aggregated result while facing the trade-off between privacy guarantee against colluding clients and robustness against dropped clients.
In \cite{FangGHMFY21}, only clients have the shared private key, and thus only clients can decrypt the aggregation result, meaning that the decryption right of the aggregated result is the clients. However, due to the shared private key, \cite{FangGHMFY21} cannot tolerate collusion between the server and clients. This means that if one client is compromised, the privacy of other clients cannot be guaranteed. Fortunately, this scheme can achieve robustness against up to $N-2$ dropped clients, which is more flexible than the other three schemes.

As described in Section \ref{sec:scheme}, the system initialization of our scheme is conducted by the server, and clients are free to choose their own public and private keys. Thus, our scheme satisfies the properties of no trusted third-party participation and client-defined private key. As analyzed in Section \ref{sec:security}, the Theorem \ref{the:aggregation} demonstrates that only the server can decrypt the aggregated result, which implies that the decryption right of the aggregation result belongs to the server. Furthermore, the Theorems \ref{the:collusion} and \ref{the:fault} show that our scheme can withstand up to $N-2$ colluded clients and tolerate up to $N-2$ dropped clients. Most notably, compared with the secret sharing-based schemes \cite{BonawitzIKMMPRS17, SoNYL0AGA22, FangGHMFY21}, our scheme does not face the trade-off between privacy and dropout-tolerant robustness.

\begin{figure*}[htbp]    
  \centering           
  \subfloat[The accuracy result on MNIST]  
  {
      \label{convergence_mnist}\includegraphics[width=0.33\textwidth]{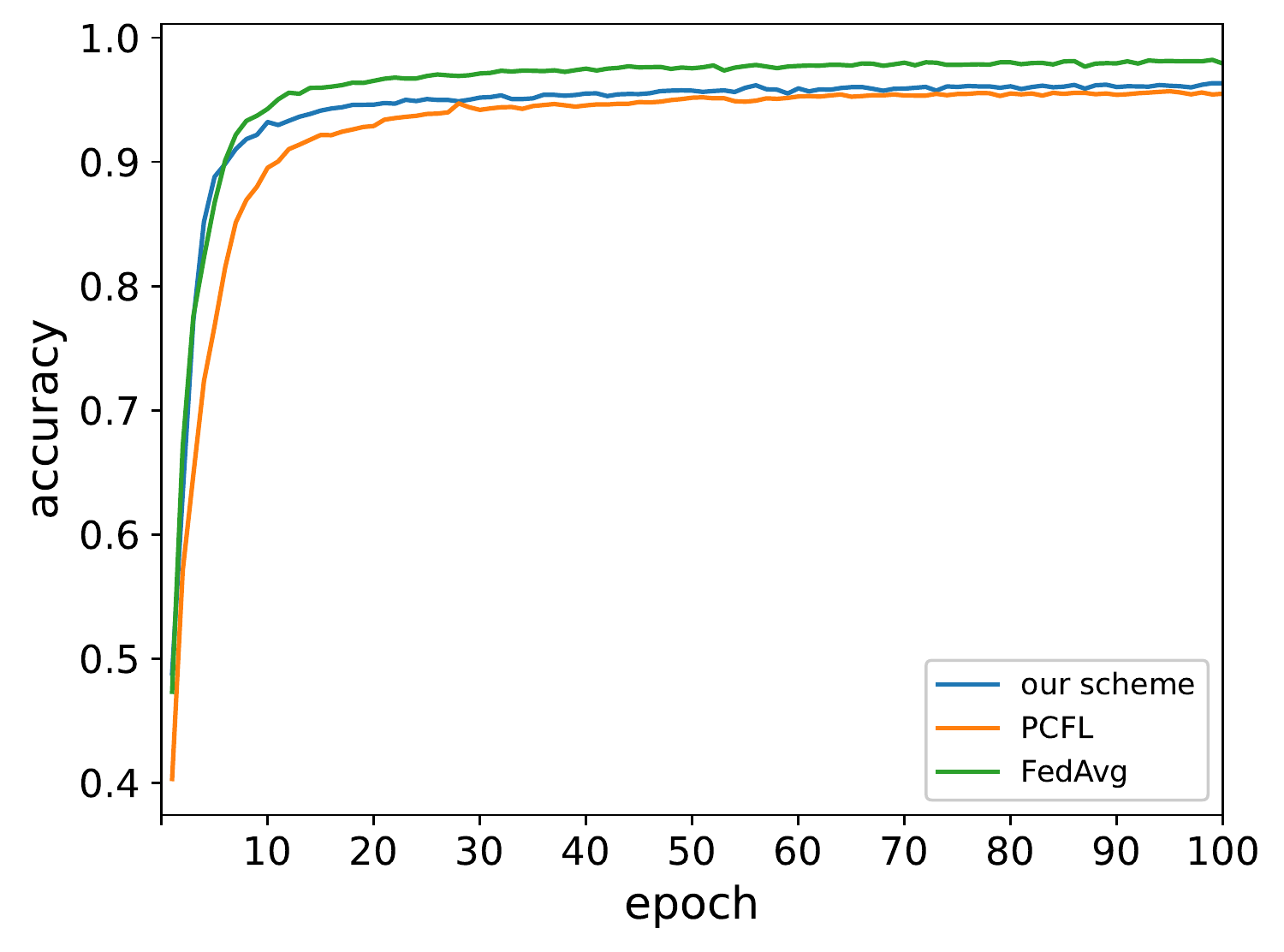}
  }
  \subfloat[The accuracy result on Cifar10]
  {
      \label{convergence_cifar10}\includegraphics[width=0.33\textwidth]{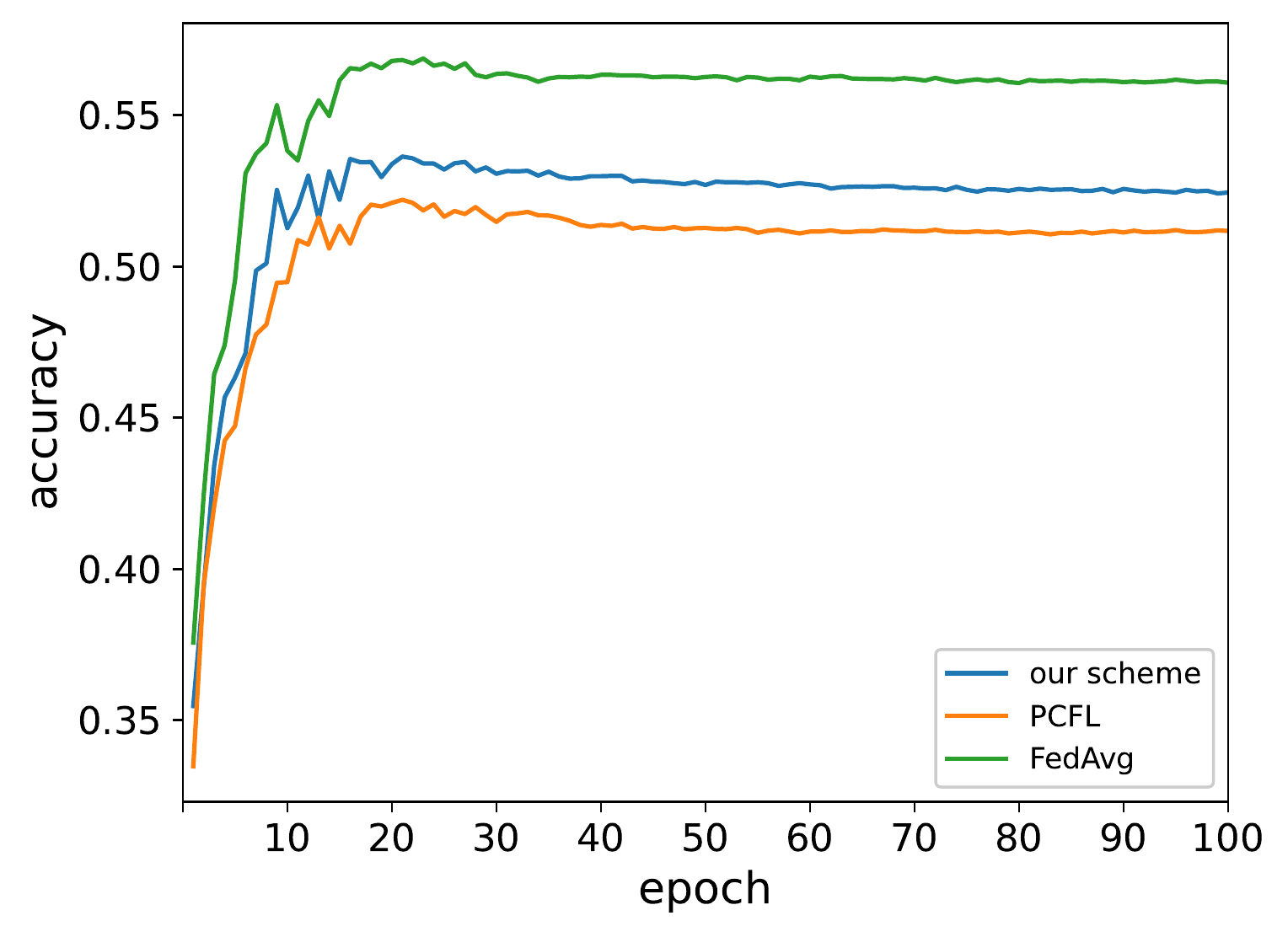}
  }
  \subfloat[The accuracy result on Cifar100]
  {
      \label{convergence_cifar100}\includegraphics[width=0.33\textwidth]{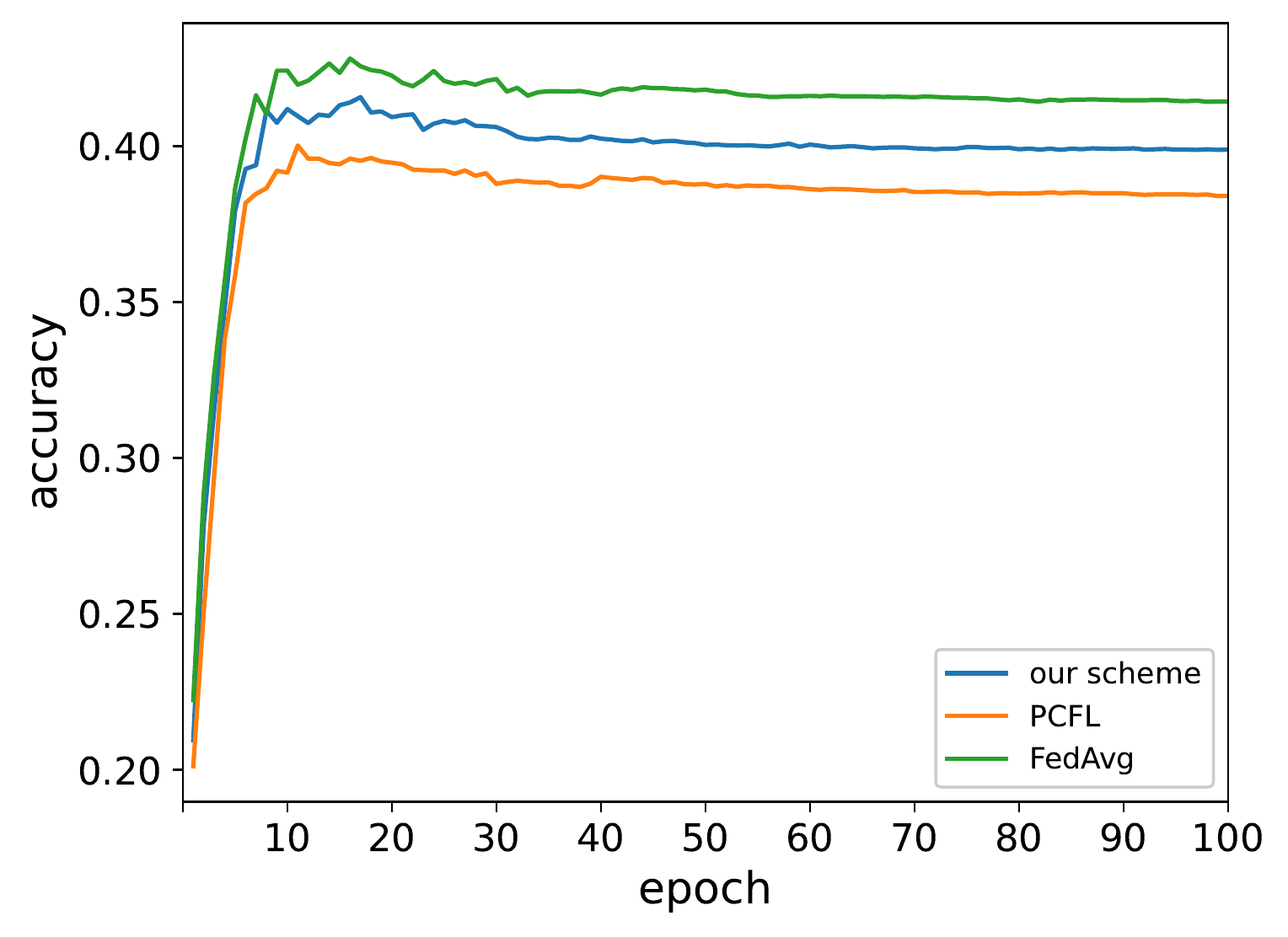}
  }
  \caption{The accuracy result on different datasets}    
  \label{convergence_rate}            
\end{figure*}

\section{Performance Evaluation}\label{sec:perfor}
In this section, we analyze our scheme through simulation and compare it with the most related work \cite{FangGHMFY21} (referred to as PCFL), as well as a non-private federated learning scheme (i.e., FedAvg \cite{McMahanMRHA17}), in terms of model accuracy, computational costs, and communication overhead. We conduct experiments using the Torch libraries \cite{PaszkeGMLBCKLGA19} on an Ubuntu system equipped with an Intel Xeon CPU and TITAIXP GPU. In the simulation, we establish a federated learning system with $10$ clients to train a CNN model using the MNIST dataset, Cifar10 dataset, and Cifar100 dataset, respectively. Prior to model upload, each client trains the local model for $5$ rounds with a learning rate of $0.01$. In our experiments, we set $\kappa_{1}=512$ bits and $\kappa_{2}=256$ bits for security requirements.

\subsection{Comparison of training accuracy}
We compare the training accuracy of our scheme to the PCFL \cite{FangGHMFY21} and FedAvg \cite{McMahanMRHA17} under different epochs in Fig. \ref{convergence_rate}. The figure shows that all three schemes converge almost within the same epoch, but the training accuracies for both our scheme and PCFL \cite{FangGHMFY21} on MNIST, Cifar10, and Cifar100 are lower than the non-private federated learning scheme FedAvg \cite{McMahanMRHA17}. It is worth noting that neither our scheme nor PCFL \cite{FangGHMFY21} theoretically loses accuracy since homomorphic encryption and secret sharing techniques are lossless privacy-preserving methods. However, by observing the figure, we can conclude that both our scheme and PCFL \cite{FangGHMFY21} do experience some loss in training accuracy. This is because local gradients are real numbers, while privacy-preserving methods are performed in an integer domain. In order to implement privacy-preserving methods successfully in secure aggregation, real numbers must be converted to integers, which inevitably results in precision loss during the rounding of real numbers \cite{MohasselZ17}. Fortunately, this accuracy loss can be minimized by designing an appropriate conversion method, which we plan to discuss in our future work.

\subsection{Comparison of computational costs}
We depict the comparison of computational costs in Fig. \ref{running_time}. From the figure, it is evident that as the number of gradients increases (i.e., ranging in size from 10,000 to 40,000), the running times of both our schemes and PCFL \cite{FangGHMFY21} gradually increase and are larger than that of FedAvg \cite{McMahanMRHA17}. Additionally, compared to PCFL \cite{FangGHMFY21}, our scheme achieves almost 2x improvement in computational costs. It is widely acknowledged that the introduction of privacy preservation inevitably sacrifices computational efficiency, which explains why the running times of both secure aggregation schemes are greater than that of non-private FedAvg \cite{McMahanMRHA17}. Notably, our scheme introduces a super-increasing sequence to compress multidimensional gradients into 1-D, significantly reducing the number of encryption and decryption operations. In essence, our scheme performs $u=\lceil\frac{n}{k}\rceil\in [1, n]$ \footnote{Note that $u$ is usually much less than $n$ since our scheme allows larger message space, resulting in a larger split interval $k$.} encryption and decryption operations, while PCFL \cite{FangGHMFY21} has to perform $n$ encryption and decryption operations.
Consequently, the computational efficiency of our scheme is much better than that of PCFL \cite{FangGHMFY21}.

\begin{figure*}[htbp]    
  \centering           
  \subfloat[Running time of each round on MNIST]  
  {
      \label{time_mnist}\includegraphics[width=0.33\textwidth]{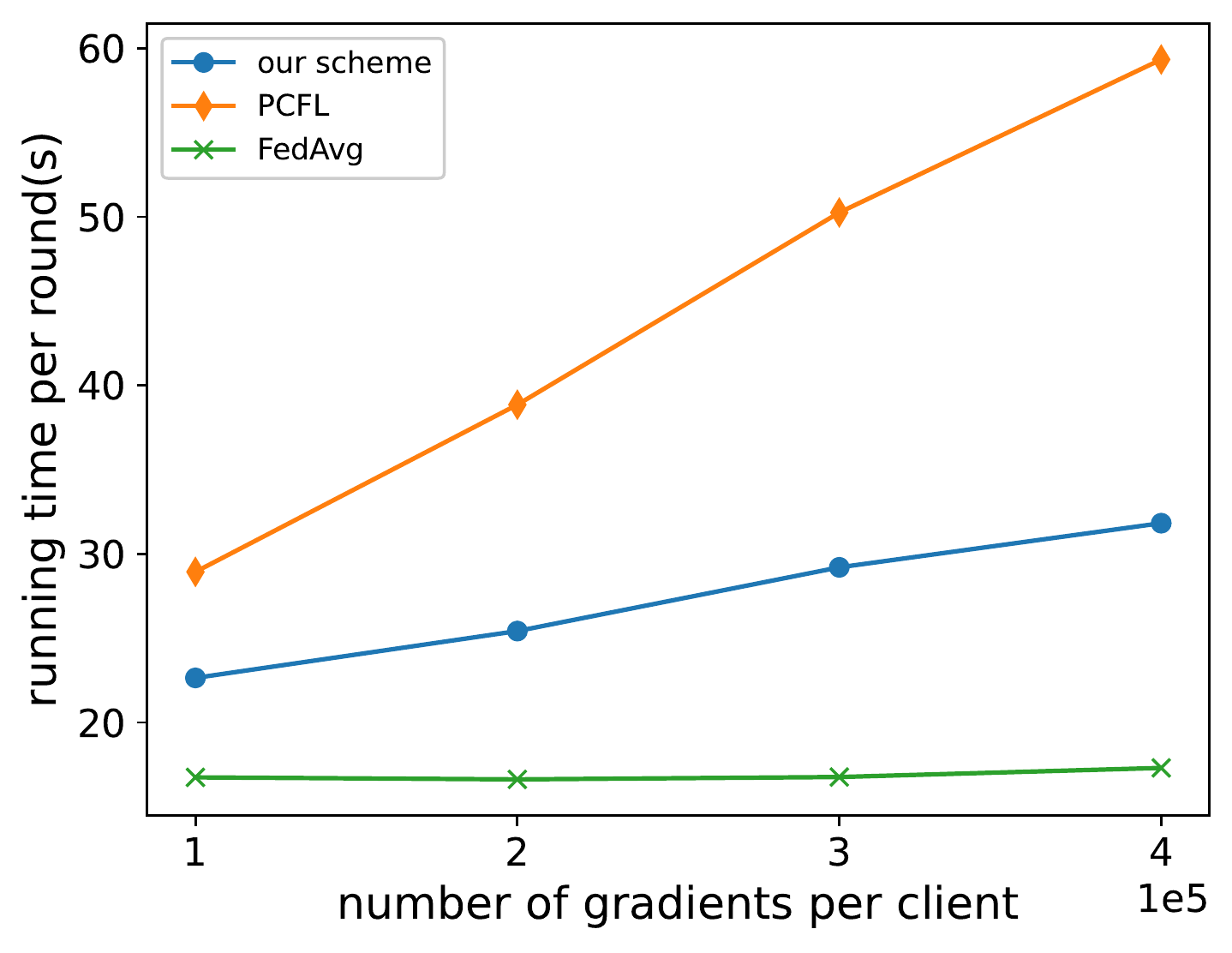}
  }
  \subfloat[Running time of each round on Cifar10]
  {
      \label{time_cifar10}\includegraphics[width=0.33\textwidth]{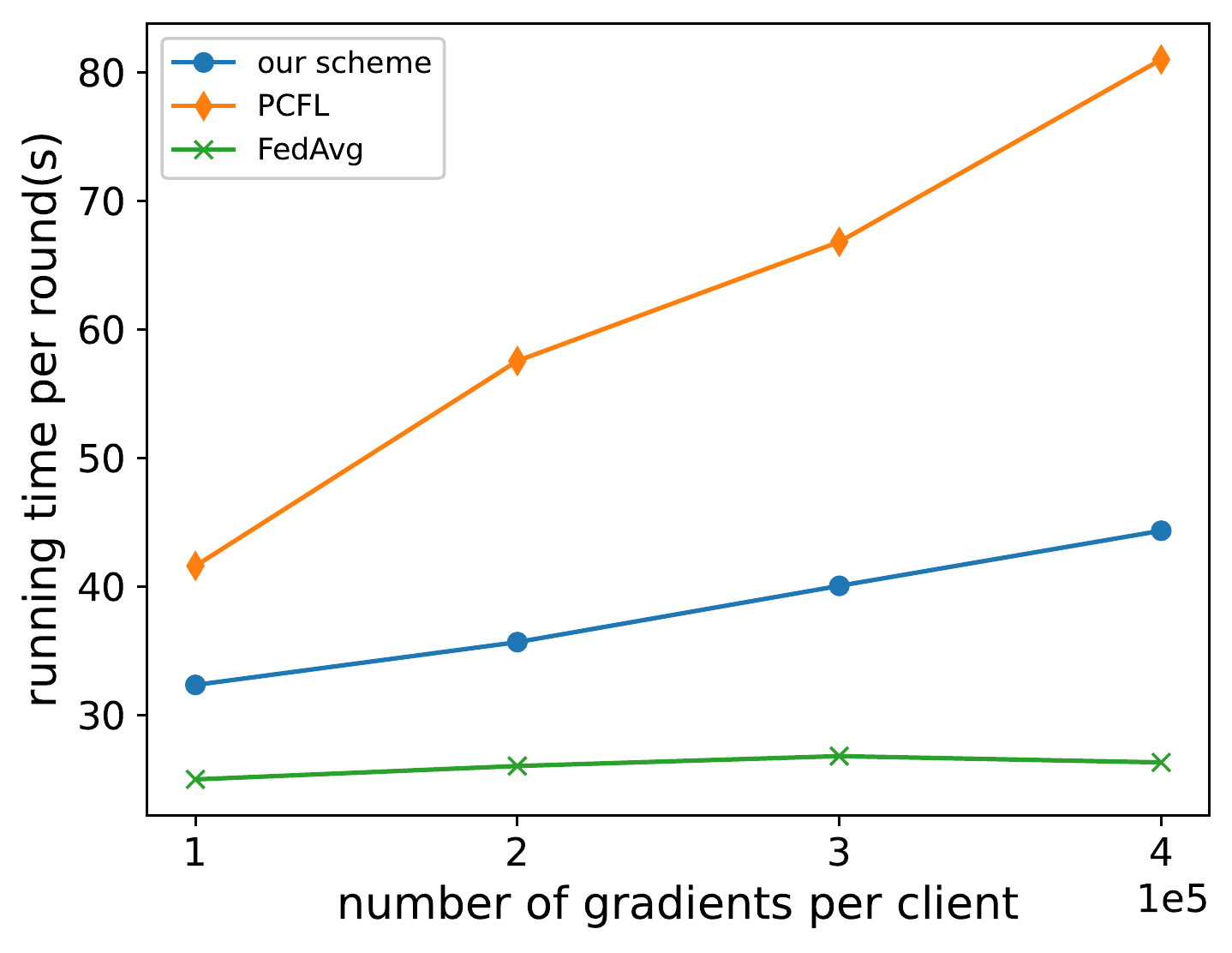}
  }
  \subfloat[Running time of each round on Cifar100]
  {
      \label{time_cifar100}\includegraphics[width=0.33\textwidth]{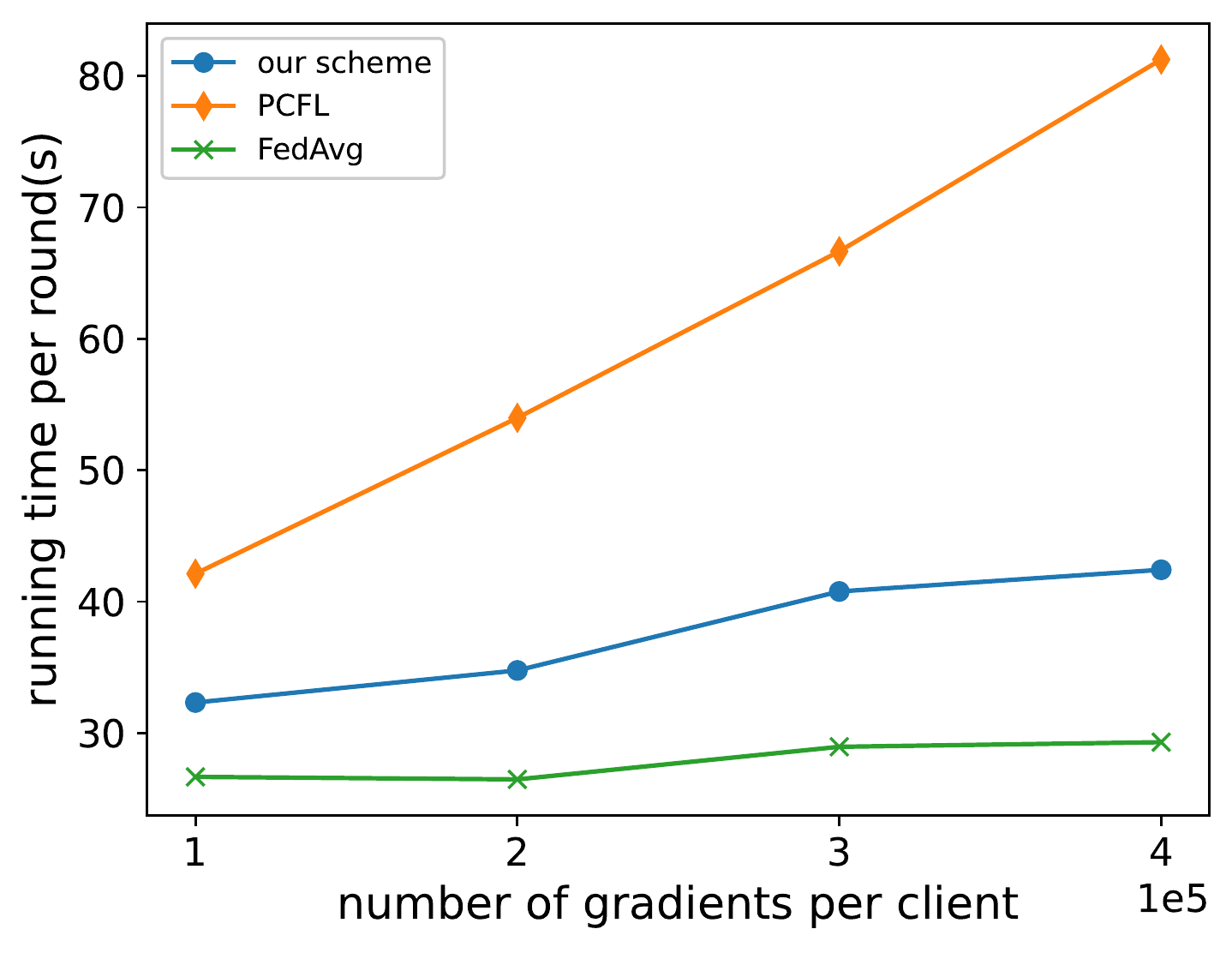}
  }
  \caption{Running time of each round}    
  \label{running_time}            
\end{figure*}

\subsection{Comparison of communication overhead}
We illustrate the comparison of communication overhead in Fig. \ref{communcation_overhead}, which demonstrates that the communication overheads of both our scheme and the non-private FedAvg \cite{McMahanMRHA17} are nearly equal and much smaller than that of PCFL \cite{FangGHMFY21}. Furthermore, the communication advantage of our scheme 
becomes even more apparent as the number of model gradients increases. The main reason also owns to the introduction of the super-increasing sequence and the expansion of the plaintext space that can be processed for encryption and decryption. 

More specifically, because of the expansion of the plaintext space, the dimension that can be compressed is larger, i.e., $k$ becomes larger. Thus, when both $n$ and $k$ increase, the growth rate of our scheme's ciphertext (i.e., the number of ciphertexts is $\lceil\frac{n}{k}\rceil$ ) is evidently lower than that of PCFL \cite{FangGHMFY21} (i.e., the number of ciphertexts is $n$).

 Overall, the efficiency advantages of our scheme will become increasingly prominent as the number of model parameters increases, as compared to the state-of-the-art homomorphic encryption-based secure aggregation schemes.

\begin{figure*}[htbp]    
  \centering           
  \subfloat[Communication overhead on MNIST]  
  {
      \label{communication_mnist}\includegraphics[width=0.33\textwidth]{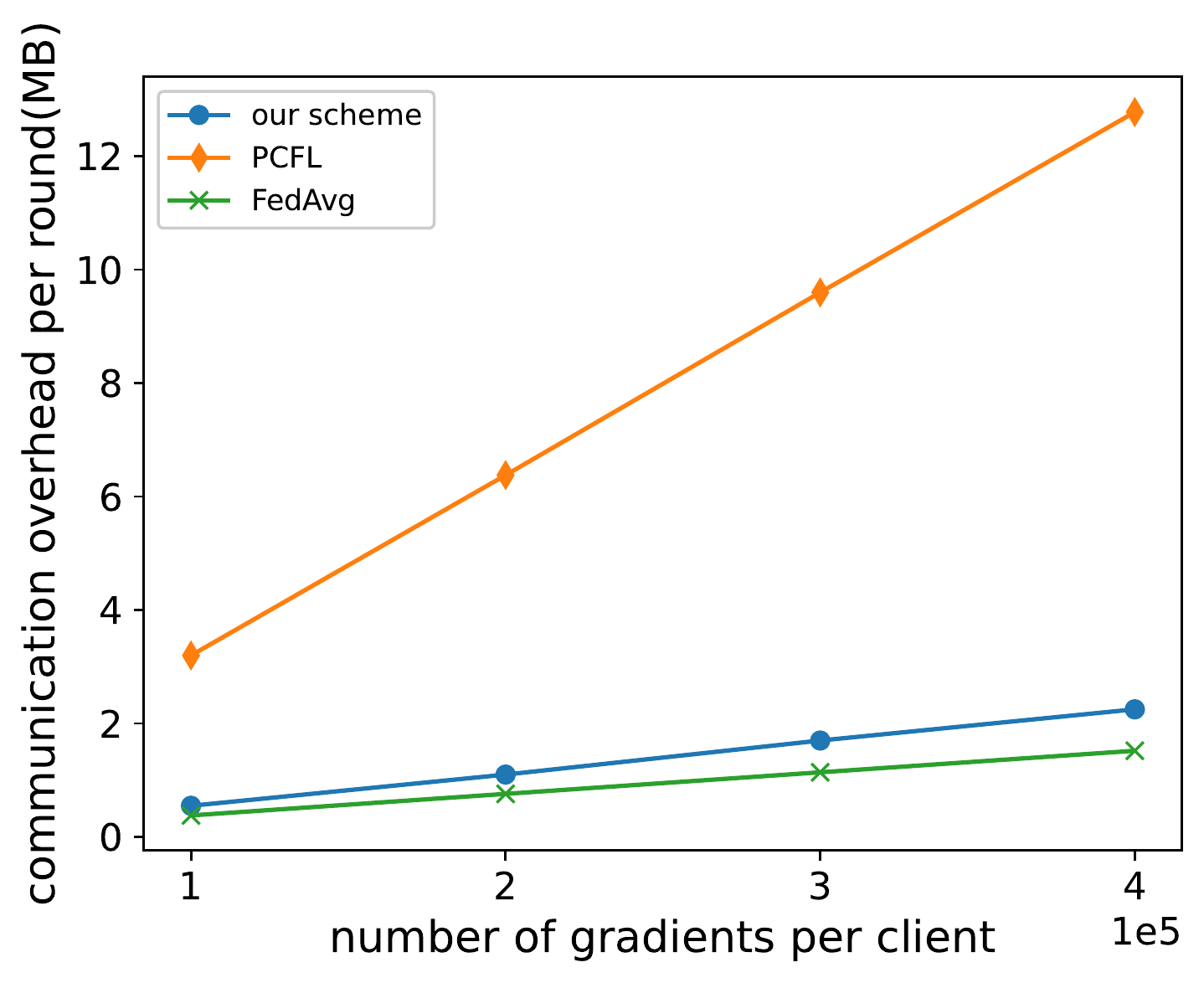}
  }
  \subfloat[Communication overhead on Cifar10]
  {
      \label{tcommunication_cifar10}\includegraphics[width=0.33\textwidth]{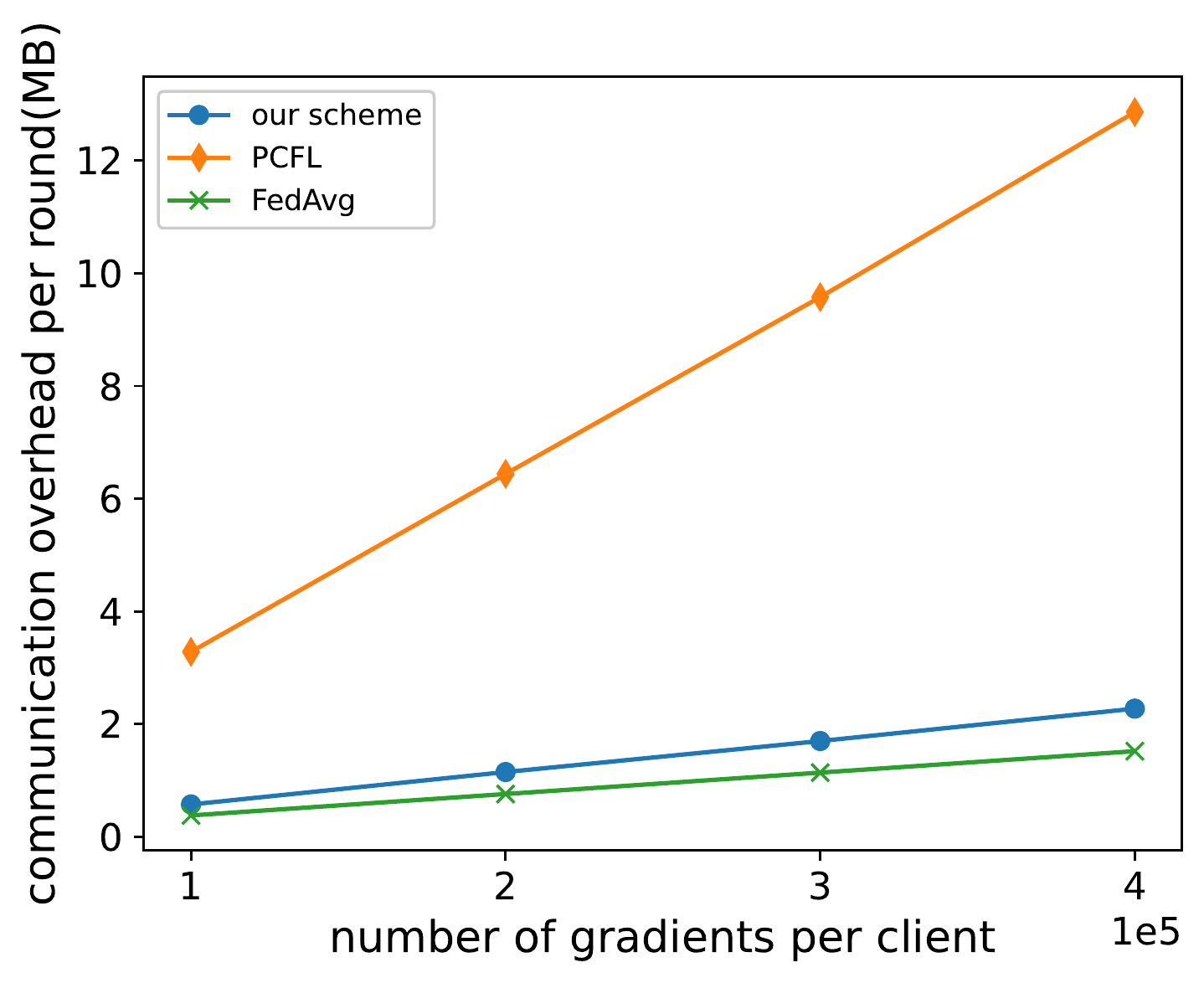}
  }
  \subfloat[Communication overhead on Cifar100]
  {
      \label{communication_cifar100}\includegraphics[width=0.33\textwidth]{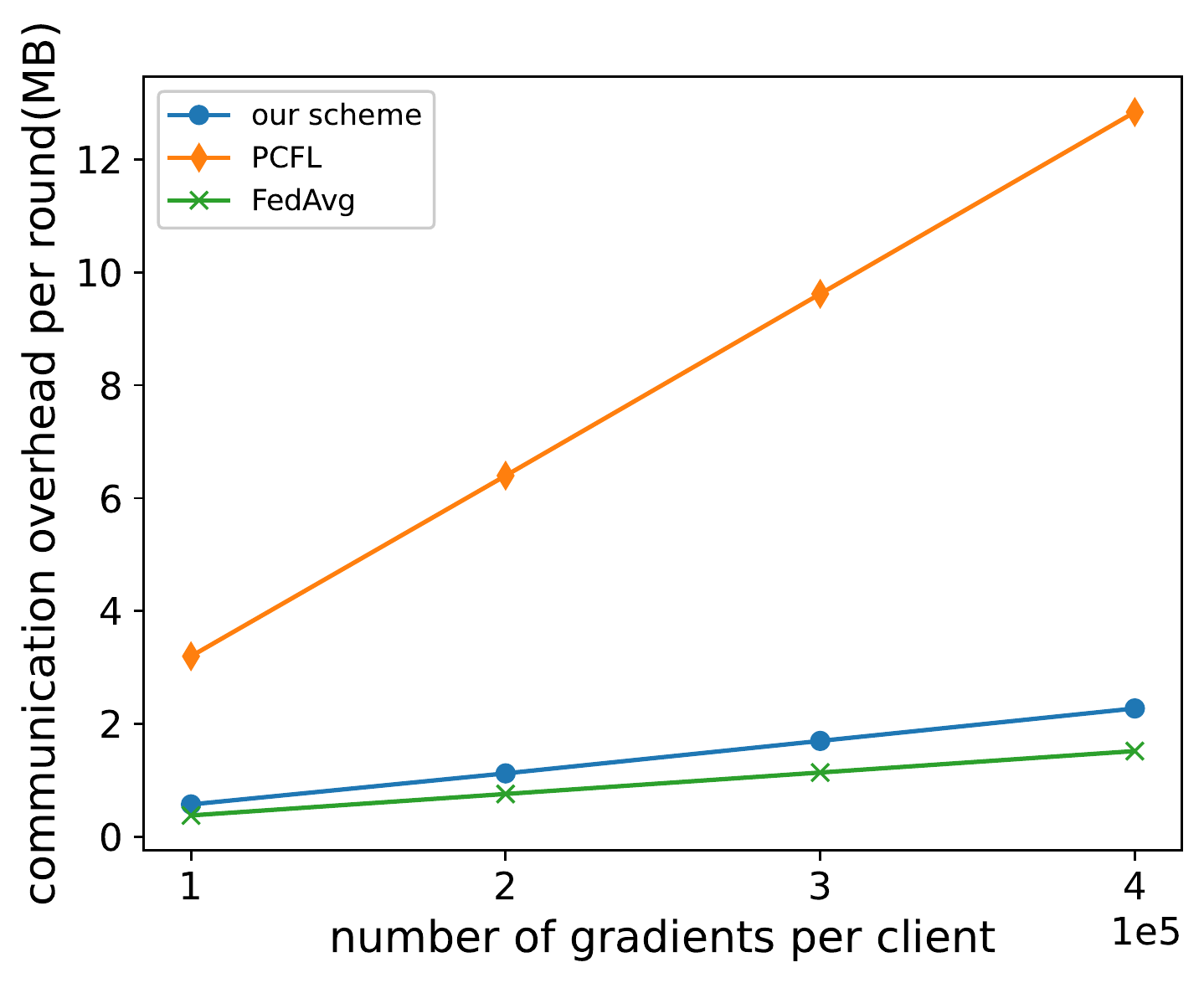}
  }
  \caption{Communication overhead of each round}    
  \label{communcation_overhead}            
\end{figure*}

\section{Related Work}\label{sec:related}
The concept of federated learning was first introduced by Google in 2016 \cite{KonecnyMYRSB16} as a solution to privacy leaks for distributed data providers. The central idea of federated learning is that, under the orchestration of a central server, numerous distributed data providers (also referred to as ``client'') collaborate locally to train a shared global model without uploading local training data to  the data center \cite{ShokriS15}. The most popular implementation of federated learning is based on iterative model averaging, known as FedAvg \cite{McMahanMRHA17}. In FedAvg, many clients locally train the model with local data to obtain the local gradient and send it to the server. The server then aggregates (e.g., calculates the weighted average of local gradients) all received local gradients to update the global model. Subsequently, various federated learning schemes \cite{ZhangXBYLG21} and implementations \cite{FATE, TensorFlow} have been presented based on the FedAvg. However, recent studies \cite{ZhuLH19, GeipingBD020} have proved that federated learning still suffers from privacy leakage. For example, \cite{ZhuLH19} showed that the adversary is capable of reconstructing the training data from the transmitted local gradient. 

To enhance the privacy preservation of each client's training data, numerous secure aggregation schemes for federated learning have been researched \cite{PhongAHWM18, PhongA0WM17, HaoLXLY19, ChaiWCY21, ZhangLX00020, FangGHMFY21, BonawitzIKMMPRS17}. The core concept behind these schemes was for each client to encrypt the local gradient before uploading it, with the server conducting the aggregation on the encrypted gradients. Consequently, the server could obtain no information other than aggregated results. More specifically, these secure aggregation schemes were primarily designed based on two cryptographic techniques: additive homomorphic encryption (HE) technique \cite{Paillier99} and $(t, n)$-threshold secret sharing technique \cite{Shamir79}. Unfortunately, these existing secure aggregation schemes had numerous drawbacks, which significantly hindered their practical application. For example, in the additive HE-based secure aggregation schemes \cite{PhongAHWM18, PhongA0WM17, HaoLXLY19, ChaiWCY21, ZhangLX00020}, all clients shared a pair of public and private keys, and encrypted local gradient with the same public key. Clearly, once a client was compromised by the adversary, this system was no longer secure. To overcome this disadvantage, \cite{FangGWJ20, HaoLLXYL20, FangGHMFY21} considered assigning different keys to different clients, allowing each client to encrypt the local gradient with its own public key using the same additive HE algorithm. These schemes employed the $(t, n)$-threshold secret sharing technique to assign different private keys to different clients, given a system private key. As long as more than $t$ clients' encrypted gradients were aggregated, the server could obtain the aggregated result. 
 However, these existing HE-based schemes required a trusted third party to complete key generation, and suffered from a significant computational cost and communication overhead, particularly for deep models with a large number of parameters.

To achieve better efficiency, several secure aggregation schemes \cite{BonawitzIKMMPRS17, DongCSW20, XuLL0L20} have solely relied on the $(t, n)$-threshold secret sharing technique. In these schemes, each client randomly selected two secret values as private keys and created additive shares for them. They subsequently distributed these additive shares among each other and used them to mask the local gradient for privacy preservation. Indeed, these schemes allowed for a maximum of $N-t$ dropped clients or $t-1$ colluded clients, which creates a trade-off between security and dropout-resiliency guarantee. Furthermore, when no more than $N-t$ dropped clients were present, the remaining online clients would have to upload additional secret shares to aid the server in recovering aggregated gradients. Unfortunately, each client was required to store $2N$ secret additive shares. Furthermore, when two secret values were chosen differently for different iterations, all clients were compelled to compute the corresponding additive shares and transmit them to one another. This operation would obviously increase the extra interactive time and communication overhead. Additionally, these secret sharing-based schemes did not consider the privacy preservation of the aggregated result. In other words, any adversary (not just the server) could obtain the aggregated result once the transmitted data was obtained.

To the best of our knowledge, it is still a challenge to develop an efficient secure aggregation scheme that offers support for client self-selected keys, robustness against collusion attacks, and robustness in tolerating dropped clients.

\section{Conclusion}\label{sec:conc}
This paper introduces an efficient and multi-private key secure aggregation scheme for federated learning. Unlike most homomorphic encryption-based schemes, our scheme does not rely on a trusted third party to initialize the system. We have skillfully designed a secure interactive protocol that enables additive homomorphic operation when clients are free to select their own public/private keys, which greatly enhances the security of the entire system. Additionally, we have thoughtfully designed encryption/decryption operations and combined them with the super-increasing vector to expand the plaintext space, thus reducing the number of encryption/decryption operations and the number of ciphertexts. Detailed security analyses show that our scheme achieves semantic security of both individual local gradients and the aggregated result, while achieving optimal robustness in tolerating client collusion and dropped clients. Extensive performance evaluations on three popular datasets demonstrate that our scheme outperforms existing competing schemes in terms of computational and communication efficiencies.

\bibliographystyle{IEEEtran}
\bibliography{ref}

\begin{IEEEbiography}[{\includegraphics[width=1in,height=1.25in,clip,keepaspectratio]{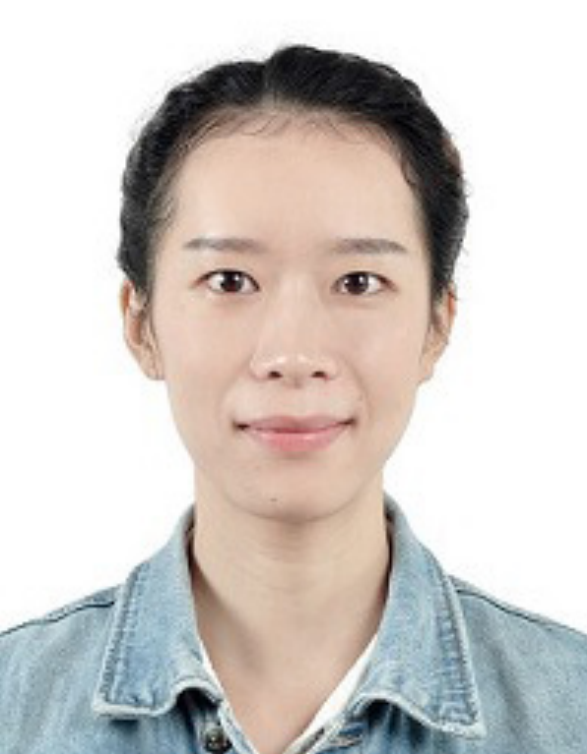}}]{Xue Yang} received the Ph.D. degree in information and communication engineering from Southwest Jiaotong University, China, in 2019. She was a visiting student at the Faculty of Computer Science, University of New Brunswick, Canada, from 2017 to 2018. She was a postdoctoral fellow with the Tsinghua Shenzhen International Graduate School from 2019 to 2021. She is currently a research associate with the School of Information Science and Technology, Southwest Jiaotong University, China. Her research interests include big data security and privacy, applied cryptography, and federated learning. 
\end{IEEEbiography}

\begin{IEEEbiography}[{\includegraphics[width=1in,height=1.25in,clip,keepaspectratio]{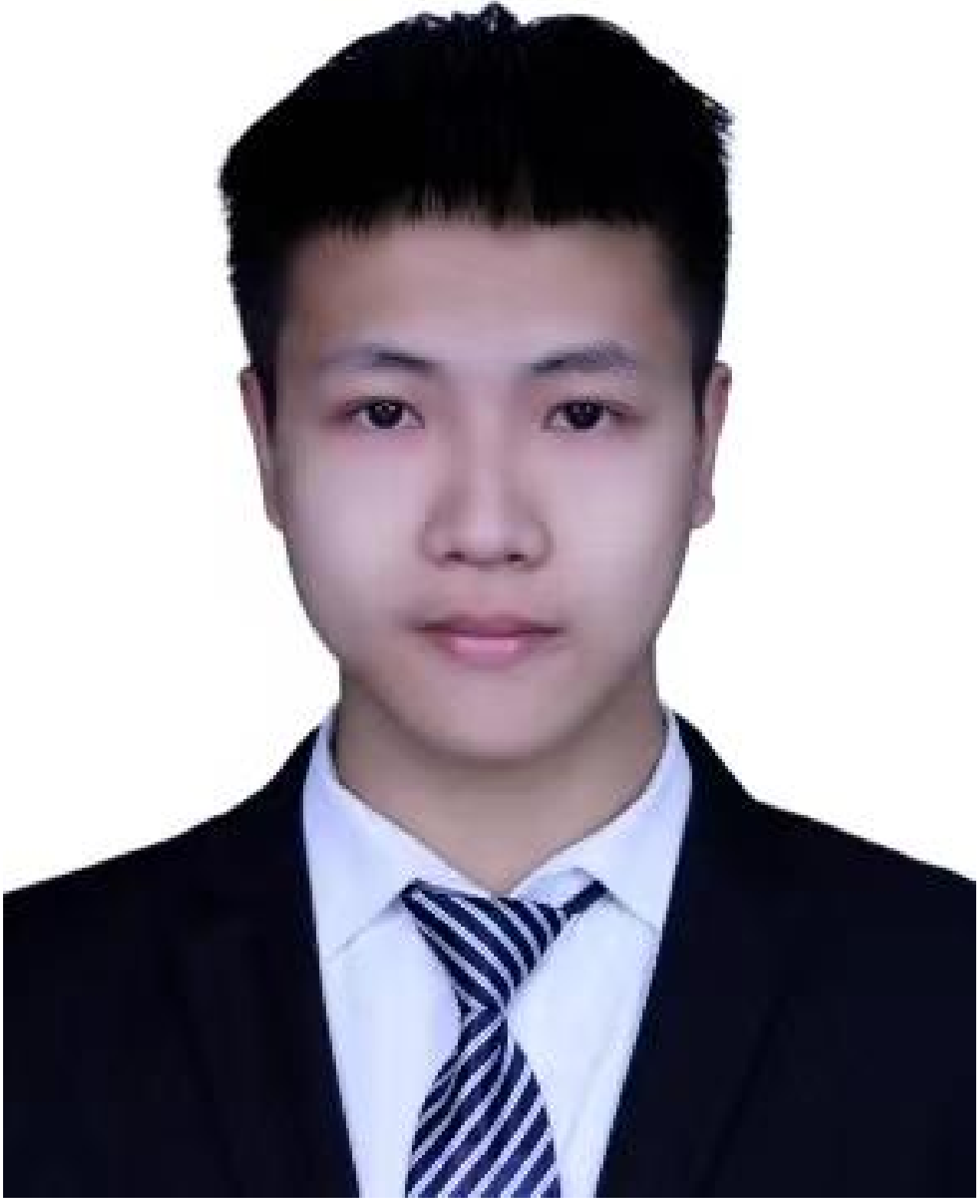}}]{Zifeng Liu} received the B.E. degree in communication engineering from Southwest Jiaotong University, Chengdu, China, in 2021. He is currently studying at the School of Information Science and Technology, Southwest Jiaotong University. His main research focuses on applied cryptography, federated learning, big data security, and privacy. 
\end{IEEEbiography}
\begin{IEEEbiography}[{\includegraphics[width=1in,height=1.25in,clip,keepaspectratio]{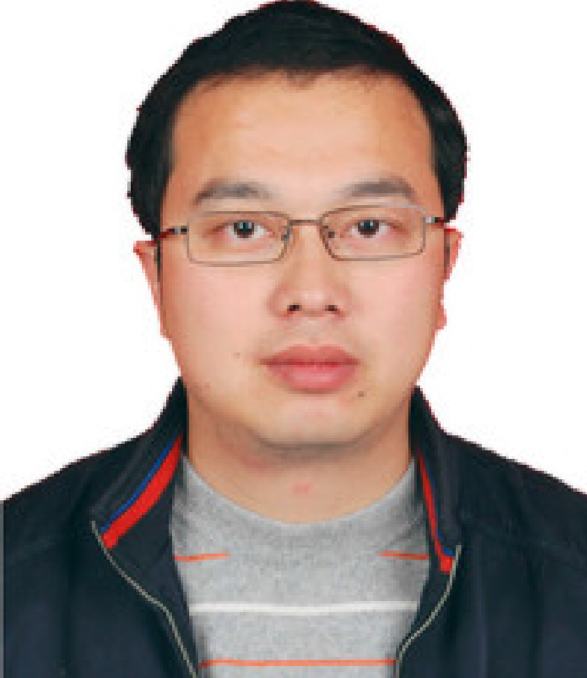}}]{Xiaohu Tang} (Senior Member, IEEE)  received the B.S. degree in applied mathematics from the Northwestern Polytechnical University, Xi'an, China, the M.S. degree in applied mathematics from the Sichuan University, Chengdu, China, and the Ph.D. degree in electronic engineering from the Southwest Jiaotong University, Chengdu, China, in 1992, 1995, and 2001 respectively. 

From 2003 to 2004, he was a research associate in the Department of Electrical and Electronic Engineering, Hong Kong University of Science and Technology. From 2007 to 2008, he was a visiting professor at University of Ulm, Germany. Since 2001, he has been in the School of Information Science and Technology, Southwest Jiaotong University, where he is currently a professor. His research interests include coding theory, network security, distributed storage, and information processing for big data.

Dr. Tang was the recipient of the National Excellent Doctoral Dissertation
Award in 2003 (China), the Humboldt Research Fellowship in 2007
(Germany), and the Outstanding Young Scientist Award by NSFC in 2013
(China). He served as Associate Editor for several journals including \textit{IEEE Transactions on Information Theory} and \textit{IEICE Transactions on Fundamentals}, and served on a number of technical program committees of conferences.
\end{IEEEbiography}
\begin{IEEEbiography}[{\includegraphics[width=1in,height=1.25in,clip,keepaspectratio]{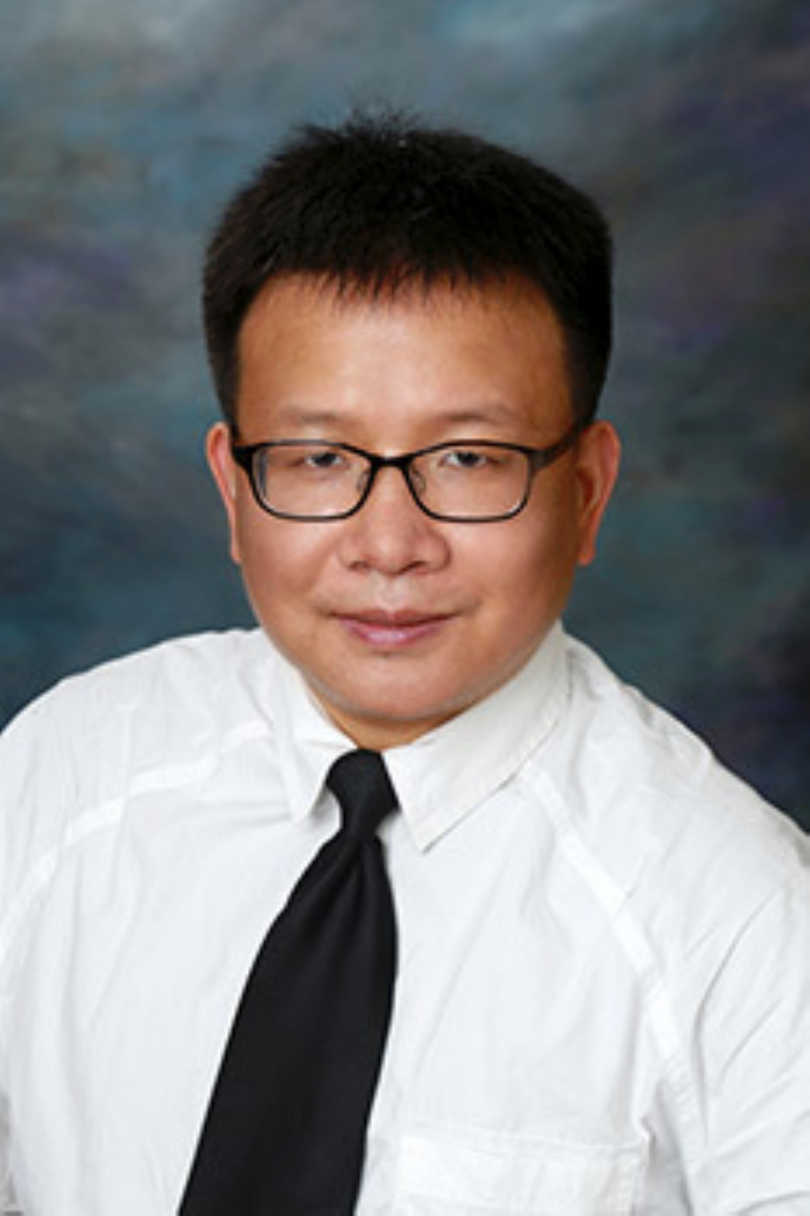}}]{Rongxing Lu} (Fellow, IEEE) received the Ph.D. degree from the Department of Electrical and Computer Engineering, University of Waterloo, Canada, in 2012. He worked as a Post-Doctoral Fellow with the University of Waterloo from May 2012 to April 2013. He is currently a Mastercard IoT Research Chair and an Associate Professor with the Faculty of Computer Science (FCS), University of New Brunswick (UNB), Canada. Before that, he worked as an Assistant Professor with the School of Electrical and Electronic Engineering, Nanyang Technological University (NTU), Singapore, from April 2013 to August 2016. His research interests include applied cryptography, privacy-enhancing technologies, and IoT-big data security, and privacy. He also serves as the Chair of IEEE ComSoc CISTC, and the Founding Co-Chair of IEEE TEMS Blockchain and Distributed Ledgers Technologies Technical Committee (BDLT-TC). 
\end{IEEEbiography}
\begin{IEEEbiography}[{\includegraphics[width=1in,height=1.25in,clip,keepaspectratio]{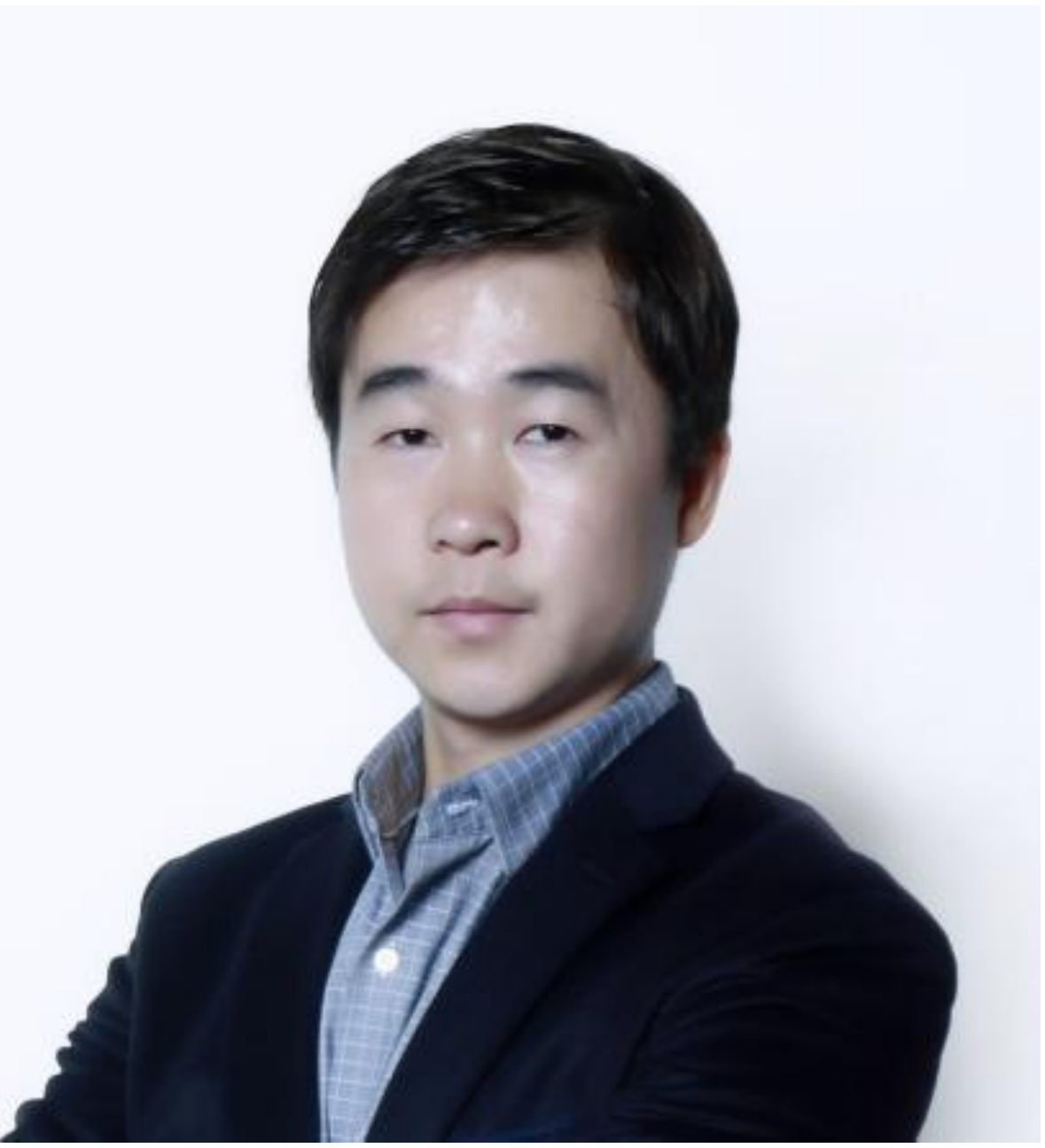}}]{Bo Liu} received the B.E. degree from Zhejiang University, Hangzhou, China, in 2006, and the Ph.D. degree in computer science from the University of Maryland, College Park, MD, USA, in 2012. He is currently the CTO of the DBAPPSecurity Company Ltd., Hangzhou. He is also the Deputy Director of the Information Security Research Center, Zhejiang Laboratory, Hangzhou, and the Deputy Director of the Joint Research Center between Zhejiang University and DBAPPSecurity Company Ltd. His research interests include cybersecurity situational awareness, data security, privacy computing, and machine learning, where he has published more than 20 technical papers with more than 8000 international citations. Dr. Liu won the leading scientific and technological achievements at the World Internet Conference in 2019.
\end{IEEEbiography}

%\newpage

%\vfill

\end{document}